\documentclass[acmsmall]{acmart}

\acmJournal{PACMPL}
\acmVolume{1}
\acmNumber{POPL}
\acmArticle{1}
\acmYear{2021}
\acmMonth{1}
\acmDOI{}
\startPage{1}
\setcopyright{none}
\usepackage{booktabs}
\usepackage{subcaption}
\usepackage{xspace}
\usepackage{listings,lstautogobble}
\usepackage{tikz}
\usepackage{caption}
\usepackage{pict2e}
\usepackage{color}
\usepackage[linesnumbered,ruled,vlined]{algorithm2e}
\usepackage{enumitem}
\usepackage{multirow}
\usepackage{adjustbox}
\usepackage{mathtools}	
\usepackage{mdframed}
\usepackage{syntax}
\let\UnderScore_

\global\mdfdefinestyle{default}{%
	usetwoside=false,%
	linecolor=black,linewidth=1pt,%
	innerleftmargin=5pt,innerrightmargin=5,%
	everyline=true
}

\newcommand\Mycomb[2][^n]{\prescript{#1\mkern-0.5mu}{}C_{#2}}

\newcommand{\tabincell}[2]{\begin{tabular}{@{}#1@{}}#2\end{tabular}}

\theoremstyle{definition}
\newtheorem{definition}{Definition}[section]
\newtheorem{theorem}{Theorem}[section]

\SetCommentSty{mycommfont}

\makeatletter
\newenvironment{btHighlight}[1][]
{\begingroup\tikzset{bt@Highlight@par/.style={#1}}\begin{lrbox}{\@tempboxa}}
	{\end{lrbox}\bt@HL@box[bt@Highlight@par]{\@tempboxa}\endgroup}

\newcommand\btHL[1][]{%
	\begin{btHighlight}[#1]\bgroup\aftergroup\bt@HL@endenv%
	}
	\def\bt@HL@endenv{%
	\end{btHighlight}%
	\egroup
}
\newcommand{\bt@HL@box}[2][]{%
	\tikz[#1]{%
		\pgfpathrectangle{\pgfpoint{1pt}{0pt}}{\pgfpoint{\wd #2}{\ht #2}}%
		\pgfusepath{use as bounding box}%
		\node[anchor=base west, fill=orange!25,outer sep=.5pt,inner xsep=0.5pt, inner ysep=0.15pt, rounded corners=1pt, minimum height=\ht\strutbox-.1pt,#1]{\raisebox{.01pt}{\strut}\strut\usebox{#2}};
	}%
}
\makeatother

\lstdefinestyle{mystyle}{
	frame=single,
	framexleftmargin=0pt,
	commentstyle=\color{green},
	keywordstyle=\color{blue}\bfseries,
	numberstyle=\tiny\color{gray},
	stringstyle=\color{purple},
	basicstyle=\tiny\ttfamily\bfseries,
	breakatwhitespace=false,         
	breaklines=false,                 
	captionpos=b,                    
	keepspaces=true,     
	numbers=none,                    
	numbersep=4pt,                  
	showspaces=false,                
	showstringspaces=false,
	showtabs=false,                  
	tabsize=2,
	language=Java,
	escapechar=\%,
	moredelim=**[is][\btHL]{`}{`},
	moredelim=**[is][{\btHL[fill=red!40]}]{@}{@},
}

\makeatletter
\newcommand\@firstoffour[4]{#1}
\newcommand\@secondoffour[4]{#2}
\newcommand\@thirdoffour[4]{#3}
\newcommand\@fourthoffour[4]{#4}

\newcommand\MyCancel[2][{0pt}{0pt}{0pt}{0pt}]{%
	\begingroup
	\sbox0{#2}%
	\edef\MyCancel@xoffsetleft {\number\dimexpr\@firstoffour#1}%
	\edef\MyCancel@yoffsetleft {\number\dimexpr\@secondoffour#1}%
	\edef\MyCancel@xoffsetright{\number\dimexpr\@thirdoffour#1}%
	\edef\MyCancel@yoffsetright{\number\dimexpr\@fourthoffour#1}%
	\edef\MyCancel@width {\number\wd0}%
	\edef\MyCancel@height{\number\ht0}%
	\setlength{\unitlength}{1sp}%
	\begin{picture}(\MyCancel@width,\MyCancel@height)
	\linethickness{1pt}
	\put(0,0){\box0}
	\color{red}
	\Line(\MyCancel@xoffsetleft,\MyCancel@yoffsetleft)
	(\numexpr\MyCancel@xoffsetright+\MyCancel@width\relax,
	\numexpr\MyCancel@yoffsetright+\MyCancel@height\relax)
	\end{picture}%
	\endgroup
}
\makeatother

\DeclareMathOperator*{\argmin}{argmin}
\newcommand{\cf}{\hbox{\emph{cf.}}\xspace}

\newcommand{\eg}{\hbox{\emph{e.g.}}\xspace}
\newcommand{\ie}{\hbox{\emph{i.e.}}\xspace}
\newcommand{\st}{\hbox{\emph{s.t.}}\xspace}
\newcommand{\wrt}{\hbox{\emph{w.r.t.}}\xspace}
\newcommand{\etc}{\hbox{\emph{etc.}}\xspace}

\newcommand{\tool}{\textsf{PATIC}\xspace}

\begin{document}

\title[Short Title]{Demystifying What Code Summarization Models Learned}

\author{Yu Wang}
\affiliation{
  \department{State Key Laboratory for Novel Software Technology, Department of Computer Science and Technology}              
  \institution{Nanjing University}            
  \city{Nanjing}
  \state{Jiangsu}
  \postcode{210023}
  \country{China}                    
}
\email{yuwang\_cs@nju.edu.cn}          

\author{Ke Wang}
\affiliation{
  \department{Security, Cryptography, and Blockchain}              
  \institution{Visa Research}            
  \city{Palo Alto}
  \state{CA}
  \country{USA}    
}
\email{kewang@visa.com}          

\begin{abstract}
Study patterns that models have learned has long been a focus of pattern recognition research. Explaining what patterns are discovered from training data, and how patterns are generalized to unseen data are instrumental to understanding and advancing the pattern recognition methods. Unfortunately, the vast majority of the application domains deal with continuous data (\ie statistical in nature) out of which extracted patterns can not be formally defined. For example, in image classification, there does not exist a principle definition for a label of cat or dog. Even in natural language, the meaning of a word can vary with the context it is surrounded by. Unlike the aforementioned data format, programs are a unique data structure with a well-defined syntax and semantics, which creates a golden opportunity to formalize what models have learned from source code. This paper presents the first formal definition of patterns discovered by code summarization models (\ie models that predict the name of a method given its body), and gives a sound algorithm to infer a context-free grammar (CFG) that formally describes the learned patterns.

We realize our approach in \tool which produces CFGs for summarizing the patterns discovered by code summarization models. In particular, we pick two prominent instances, code2vec and code2seq, to evaluate \tool. \tool shows that the patterns extracted by each model are heavily restricted to local, and syntactic code structures with little to none semantic implication. Based on these findings, we present two example uses of the formal definition of patterns: a new method for evaluating the robustness and a new technique for improving the accuracy of code summarization models. 

Our work opens up this exciting, new direction of studying what models have learned from source code.

\end{abstract}
\maketitle

\section{Introduction}

Riding on the major breakthroughs in deep learning together with the ever-increasing public datasets and computation power, machine learning models have enabled state-of-the-art solutions to a wide range of problems including image classification~\cite{krizhevsky2012imagenet,touvron2019fixing}, machine translation~\cite{devlin-etal-2019-bert,NIPS2019-8928}, and game playing~\cite{Silver2016,Silver2017,Vinyals2019}. 

The success of the learning-based approaches can be largely attributed to their capability of discovering patterns exhibited in a large amount of data. An influential subfield within machine learning has been dedicated to explaining, visualizing patterns that models have learned from data. Moreover, the field has been receiving growing attention for its leading role in tackling some of the most imminent challenges in Artificial Intelligence (AI). For example, explainability is likely to be a central goal of the next-generation AI technology.
Revealing what models have learned is a crucial first step to designing such explainable AI systems.
In addition, from a scientific standpoint, dissecting the internal operation and behavior of complex models is necessary. Because without a clear understanding of how and why machine learning models work, the development of better models is reduced to trial-and-error. 

For a few notable efforts,~\citet{chen2006composite} present a context-sensitive grammar to model
the wide variations in object configurations via composite graphical templates. 
A strength of their approach is the ability to explain what patterns are recognized from test data during the inference time, in the case of cloth modeling discussed  in~\cite{chen2006composite}, predictions are made using templates representing shoes, hands, faces, \etc~\citet{10.1007Z} introduce a visualization technique that gives insights into the function of individual feature layers and the end-to-end operation of a convolutional network, one class of Deep Neural Network (DNN) commonly applied in image classification. As a diagnostic tool, the visualization technique allows them to find model architectures that outperform AlexNet~\cite{krizhevsky2012imagenet}, the then state-of-the-art model on ImageNet~\cite{deng2009imagenet}. We defer a detailed survey of related work to Section~\ref{sec:rel}.

Despite the significant stride, formalizing patterns that models have learned remains to be an exceedingly challenging task. This is in large part due to the nature of the problem domain to which learning-based approaches are applied. Models almost exclusively deal with continuous data out of which learning formal patterns is difficult if not impossible. For example, no machine learning models known to this day set out to learn principle definitions for labels in ImageNet (\eg panda, ostrich, goldfish, \etc). 
Even in the area of natural language processing, learning formal patterns is a tricky task since the meaning of words can be ambiguous. 

When programming seems to have become yet another popular domain for machine learning models (exemplified by the DNNs), it is vitally important to recognize that program is a fundamentally different data structure. Specifically, it is discrete in nature with a well-defined syntax and semantics. Syntactically, programs are written in a way that satisfies the recursive rules (\ie production rules) defined in a context-free grammar. Semantically, the behavior of a program satisfies the inference rules defined in the small-step semantics~\cite{plotkin87structural}. All of the above lead to the insight of our work, that is, patterns learned by models from source code can be formalized. However, we face an important challenge: how to efficiently navigate through an enormous search space of 
diverse program properties ranging from syntax to semantics? 

Our solution is based on a key and rather unexpected observation we made about the behavior of many prominent models of code. Built on the work of ~\citet{wang2019coset}, which finds syntactically trivial and semantically preserving code edits frequently cause models to alter their predictions, we observe an even more surprising phenomenon. That is when models are given a program to predict, almost always the program can be reduced to very few statements (\ie $\leq$ 2) for which models make the same prediction as they do for the original program. This is a significant finding in two ways. First, (1) it indicates a small, local window of code sufficiently covers the patterns that models look for to predict the properties of the entire program.
Therefore, the space for searching the pattern definitions is orders of magnitude smaller than one would have anticipated. Second, (2) for such simple patterns which are often semantically meaningless, predicates of semantic properties can be safely ignored, which further restricts the search space to predicates of syntactic properties. Based on (1) and (2), a natural idea for defining a pattern emerges: synthesizing rules based on syntactic properties of the key statements in the original program. However, there is a caveat: can the key statements alone always preserve the label models predicted for the original program regardless of the surrounding context? To address this issue, we find the set of valid programs in which the key statements do preserve the original predicted label from which we define the pattern that models learned.

At the technical level, we propose "Abstract, Mutate, Concertize, and Summarize", a novel method for pattern formalization. First, given a set of programs $\mathcal{P}$ all with a label $L$ predicted by a model $M$, we abstract away the statements on each program in $\mathcal{P}$ that do not cause $M$ to alter its prediction. We call each remaining code snippet a \textit{seed}, 
which 
captures the essence of the prediction made by $M$ for the label $L$. Second, since machine learning models have been known for their generalization capability, we conjecture programs that resemble a seed are likely to be predicted with the same label $L$. Therefore, we mutate each seed to obtain additional code snippets, which we call \textit{mutants}.
Third, we synthesize full-fledged programs by inserting statements into each seed and mutant.
In particular, we enumerate a diverse set of statements and expressions using the grammar of the language $\mathcal{P}$ are written in. Later, we pass each synthesized program to the model to get a label. Finally, we infer a context-free grammar that describes all synthesized programs for which $M$ predicts the label $L$
as the principle definition of a pattern learned by the model $M$ \wrt the label $L$.

We develop a tool, \tool, as an implementation of the method "Abstract, Mutate, Concertize, and Summarize", which automatically formalizes the patterns learned by two code summarization models: code2vec~\cite{Alon2019code2vec} and code2seq~\cite{alon2018code2seq}. Code summarization refers to a task in which models aim to infer the name of a method given its body. Figure~\ref{fig:met} shows an example. The correct prediction for this method is \texttt{reverseArray}. We target code summarization models due to the tremendous impact they have made to the programming language community. Since the publication of code2vec in \textit{Principles of Programming Languages (POPL)} two years ago, it has not only gathered many citations (\ie 130+) but also led to several interesting follow-up works (\eg code2seq, sequence GNN~\cite{fernandes2018structured}, and LiGER~\cite{Wang101145}). 

\begin{figure*}[thb!]
	\centering
	\begin{subfigure}{0.65\textwidth}
		\lstset{style=mystyle}
\begin{lstlisting}[linewidth=6.8cm,morekeywords={var, public, String}]
public String[] @f@(String[] array)
{
    String[] newArray = new String[array.Length];
    for (int index = 0; index < array.Length; index++)
        newArray[array.Length-index-1] = array[index];

    return newArray;
}
\end{lstlisting}
	\end{subfigure}
	\quad
	\begin{subfigure}{0.253\textwidth}
		\centering
		\includegraphics[width=\textwidth]{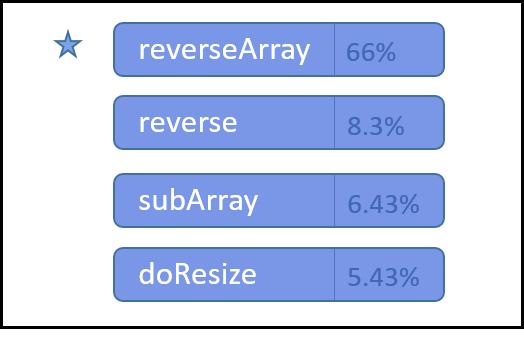}
	\end{subfigure}
	\caption{An example method extracted from~\cite{Alon2019code2vec} whose name is stripped out for models to predict. The top-4 prediction made by code2vec is shown on the right.}
	\label{fig:met}
\end{figure*}

Through the context-free grammars that \tool inferred, we find that the patterns learned by both evaluated models are simple \wrt all labels included in all Java-small, Java-med, and Java-large~\cite{alon2018code2seq}, three public, large-scale, cross-projects datasets used in many code summarization works. The average number of tokens presented in a seed is computed to be 15.59.
For almost all the methods, the vast majority of the statements in their body can be removed, and the resulted seeds consist of less than two statements.
In addition, there is little constraint on the synthesis of full-fledge programs given the presence of seed statements, many constructs we insert into seeds do not even remotely resemble the semantics of the non-seed statements. Nevertheless, models display a strong tendency to keep the predictions they made for the original programs. Our finding
indicates that neither code2vec nor code2seq tries to learn a global, semantic representation for a given method, instead, they use local, syntactic program features as a proxy to simplify their memorization of the method.


Based on our findings in what code summarization models have learned, we present two example uses of our context-free grammar-based pattern formalization. First, we propose a new method to evaluate the robustness of code summarization models. In particular, we construct attacks to expose their vulnerabilities to small perturbations to the input programs. Our intuition is to concentrate changes on the seed statements, the part of a method on which models predominately based their predictions, to sway the predictions models made for the original programs. In addition, we introduce four semantically-preserving program transformations which enable us to find very small perturbations (\ie $\leq$ 2 tree-edit distance between the ASTs of perturbed and unperturbed programs) for every correctly-predicted test method in Java-small, Java-med and Java-large.
Second, we propose a new technique to improve the generalizability of code summarization models. In the spirit of adversarial training~\cite{goodi}, we opt to include programs synthesized to address a particular weakness of code summarization models to re-train the models. Technically, after collecting the training programs for which models made the incorrect predictions, we inject their seeds to other programs to synthesize additional training data. Specifically, by assigning the label of the hosting program to the synthesized program, we guide models to shift their attention to different syntactic structures than they previously attended to, which paves the way for them to connect to the ground truth. After undergoing such a re-training process, both code2vec and code2seq have become more accurate, especially code2seq which achieves the state-of-the-art results on Java-med, and Java-large. 


This paper makes the following contributions:
\begin{itemize}
	\item A formal definition of patterns learned by code summarization models.
	
	\item A sound algorithm for formalizing patterns that code summarization models learned. Assuming the monotonicity property (\cf Definition~\ref{def:mono}), the algorithm is also complete.  
	
	\item An implementation, \tool, which automatically generates definitions of patterns discovered by code summarization models in the form of context-free grammars.

	\item An empirical evaluation of \tool on two prominent code summarization models: code2vec, and code2seq. Through the context-free grammars that \tool inferred, we find the patterns neither model learned precisely capture the properties of programs with the predicted label.
	
	\item A new method for evaluating the robustness of a model, which finds adversarial examples with smaller perturbations and within far few attempts than prior approaches.
	
	\item A new technique for improving the accuracy of code summarization models, which enables code2seq to achieve the state-of-the-art results on Java-med, and Java-large.	
\end{itemize}

\section{Overview}
\label{sec:over}

In this section, we present an overview of our approach to formalizing patterns learned by code summarization models.

\subsection{An Illustrative Example}
First, we introduce a code summarization model and two input methods as our running example for illustrating the key idea and high-level steps of our approach.

\vspace*{4pt}
\noindent
\textbf{\textit{Model.}}\, We use code2seq, the state-of-the-art code summarization model on Java-med and Java-large. Like many other DNN-based models, code2seq strives for learning precise vectorial representations for source code. Such vectors, commonly known as program embeddings, capture the semantics of a program through their numerical components such that programs denoting similar semantics will be located in close proximity to one another in the vector space. At high level, code2seq adopts a generative approach for method name prediction. It employs a standard encoder-decoder architecture~\cite{devlin2014fast, cho-cho2014learning} in which the encoder first embeds the ASTs of input methods into vectors, then the decoder uses the vectors to generate method names as sequences of words (\eg \texttt{reverse}, and \texttt{array} as a prediction in Figure~\ref{fig:met}).

\begin{figure*}[thb!]
	\centering
	\begin{subfigure}{0.472\textwidth}
		\lstset{style=mystyle}
\begin{lstlisting}[linewidth=6.4cm,morekeywords={var, public, String}]
public void saveBitmapToFile(Bitmap bmp, File 
    file) {
        
    try {
        BufferedOutputStream bos = new BufferedOut-
            putStream(new FileOutputStream(file));
        @bmp.compress(Bitmap.CompressFormat.PNG,@
            @90, bos);@
            
        bos.close();
        Log.e("succeed");
    }
    catch (IOException e) {
        Log.e(TAG, "failed to save frame", e);
    }

}
\end{lstlisting}
		\caption{}
		\label{fig:exasavebit}
	\end{subfigure}
	\begin{subfigure}{0.475\textwidth}
		\lstset{style=mystyle}
\begin{lstlisting}[linewidth=7.25cm,morekeywords={var, public, String}]
public void saveBitmapToFile(String fileStr,  Bitmap 
    bitmap, View container, boolean isShare) { 
    File file = new File(fileStr); 
    if (isShare && file.exists()) return; 
    try {
        FileOutputStream fos = new FileOutputStream(file);
        boolean isCompress = @bitmap.compress(Bitmap.@
            @CompressFormat.PNG,90, fos);@
        if (isCompress)
            SnackbarUtil.showShort(container, 
                "save to file succeeded");
        else
            SnackbarUtil.showShort(container, 
                "save to file failed");
        fos.flush(); fos.close(); }
    catch (IOException e) { e.printStackTrace(); }
}
\end{lstlisting}
		\caption{}
		\label{fig:newsavebit}
	\end{subfigure}
	\caption{Two methods with the name \texttt{saveBitmapToFile}.}
	\label{fig:exmethods}
\end{figure*}

\vspace*{4pt}
\noindent
\textbf{\textit{Input Methods.}}\, Figure~\ref{fig:exasavebit} and~\ref{fig:newsavebit} depict two Java methods with the name \texttt{saveBitmapToFile}, which are extracted from the training set of Java-large. The distinguishing feature of this label is the \texttt{compress} API under the \texttt{Bitmap} class, which both methods have highlighted in the shadow box. 
Worth noting that neither code2vec nor code2seq requires input programs to compile so long as they satisfy the syntactic grammar of the language they are written in. Also, both code2vec and code2seq only take individual methods as input. When another method is invoked in the body of the input method, no inter-procedure analysis is performed, neither is inlining. Given the methods in Figure~\ref{fig:exasavebit} and~\ref{fig:newsavebit}, code2seq gives the correct predictions for both of them. Hereinafter, when referring to the inputs of code2vec or code2seq, we use programs and methods interchangeably.  

\subsection{Overview of ``Abstract, Mutate, Concretize, and Summarize''}
We now give an overview of our technique which consists of four major steps: abstract, mutate, concretize, and summarize.

\subsubsection{Abstract}
While DNN-based models have been gaining increasing popularity in the programming domain,~\citet{wang2019coset} cautioned they are notably unstable with their predictions. Simple, natural, semantically-preserving transformations frequently cause models to change their predictions. Figure~\ref{fig:fac} depicts an example in which the original method (\ref{fig:before}) is correctly predicted to be \texttt{factorial} by code2vec, and the transformed method (\ref{fig:after}), albeit semantically equivalent, is totally mishandled. None of the top-5 predictions even remotely resembles the ground truth considering that we only swapped the operands of the multiplication.


\begin{figure*}[thb!]
	\centering
	\begin{subfigure}{0.475\textwidth}
		\begin{subfigure}{0.5\textwidth}
			\lstset{style=mystyle}
\begin{lstlisting}[linewidth=3.2cm,morekeywords={var, public, String}]
int f(int n)
{
    if (n == 0)
    {
        return 1;
    }
    else
    {
        return @n * f(n-1)@;
    }
}
\end{lstlisting}
		\end{subfigure}
		\begin{subfigure}{0.5\textwidth}
			\centering
			\includegraphics[width=\textwidth]{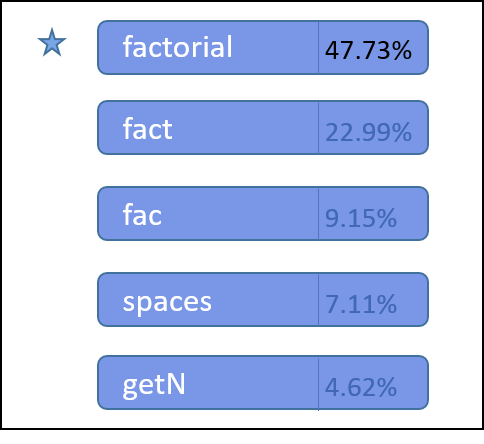}
		\end{subfigure}	
		\caption{Prediction for the original method.}
		\label{fig:before}
	\end{subfigure}
	\,
	\begin{subfigure}{0.475\textwidth}
		\begin{subfigure}{0.5\textwidth}
			\lstset{style=mystyle}
\begin{lstlisting}[linewidth=3.2cm,morekeywords={var, public, String}]
int f(int n)
{
    if (n == 0)
    {
        return 1;
    }
    else
    {
        return @f(n-1) * n@;
    }
}
\end{lstlisting}
		\end{subfigure}
		\begin{subfigure}{0.5\textwidth}
			\centering
			\includegraphics[width=\textwidth]{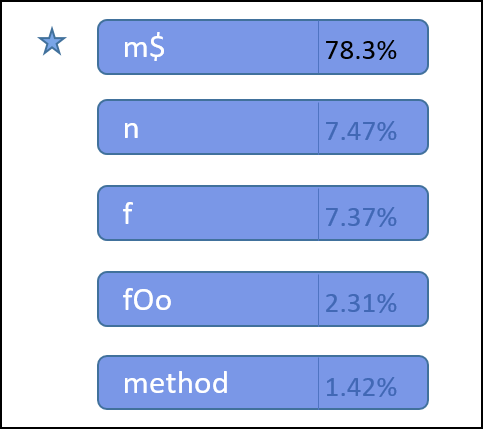}
		\end{subfigure}	
		\caption{Prediction for the transformed method.}
		\label{fig:after}	
	\end{subfigure}
	\caption{A simple, natural, semantically-preserving transformation causes code2vec to change its prediction. Note that the probability of the top-1 prediction is even higher on the transformed method.}
	\label{fig:fac}
\end{figure*}

Their finding suggests that models don't evenly distribute their attention across the entire structure of the method, instead, they focus on a small, local window of code for making predictions. To validate this hypothesis, we aim to find the window of code which models predominately based their predictions on. A simple idea is to exhaust all subsets of the statements in a given method to find 
the minimal subset for which models make the same prediction as they do for the original program. However, a challenge arises: since the number of subsets to be traversed grows exponentially with the size of the method, how does it effectively scale to methods which consist of a large number of statements. We defer a detailed discussion on how to overcome this challenge in Section~\ref{subsubsec:abs}.



Figure~\ref{fig:exmethodsseed} depicts the minimal programs we discovered for the methods in Figure~\ref{fig:exasavebit} and~\ref{fig:newsavebit}. Three points we intend to emphasize. First, (1) no statement in either seed reflects the name of the methods, \texttt{bos.close()}, albeit indicating a high probability of file operations, do not represent ``\texttt{save} \texttt{to} \texttt{file}'', neither does it connect to ``\texttt{bitmap}''. The rest are the log/display APIs which are completely irrelevant. The reason that code2seq still predicts the label \texttt{saveBitmapToFile} for both methods is the parameters provided in the method headers since changing the class name \texttt{Bitmap} to \texttt{Image} leads code2seq to predict a different label for both methods. Second, (2) the distinguishing features (\ie \texttt{compress} API under the \texttt{Bitmap} class) are absent in both seeds, casting serious doubts on code2seq about what it has learned. Finally, (3) Despite their irrelevance, all statements in both seeds are necessary in keeping code2seq's prediction. This suggests that code2seq takes into account the syntactic structure of methods --- two consecutive method invocations in the case of Figure~\ref{fig:exasavebitseed} --- to make prediction.

\begin{figure*}[thb!]
	\centering
	\begin{subfigure}{0.37\textwidth}
		\lstset{style=mystyle}
\begin{lstlisting}[linewidth=5.6cm,morekeywords={var, public, String}]
public void saveBitmapToFile(Bitmap bmp, File
    file) {
    bos.close();
    Log.e(TAG, "failed to save frame", e);
}
\end{lstlisting}
		\caption{}
		\label{fig:exasavebitseed}
	\end{subfigure}
\,\,
	\begin{subfigure}{0.39\textwidth}
		\lstset{style=mystyle}
\begin{lstlisting}[linewidth=6.3cm,morekeywords={var, public, String}]
public void saveBitmapToFile(String fileStr,  Bit-
    map bitmap, View container, boolean isShare) {
    SnackbarUtil.showShort(container, 
        "save to file succeeded");
}
\end{lstlisting}
		\caption{}
		\label{fig:newsavebitseed}
	\end{subfigure}
	\caption{(\subref{fig:exasavebitseed})/(\subref{fig:newsavebitseed}) depicts the seeds for the method in Figure~\ref{fig:exasavebit}/\ref{fig:newsavebit}.}
	\label{fig:exmethodsseed}
\end{figure*}

\subsubsection{Mutate}
It is well-known that machine learning models don't just create a mapping from input data to predicted labels through rote memorization but rather they discover patterns from data which generalize to that of similar characteristics. Based on this knowledge, we intend to find additional programs that are similar to the seeds which otherwise can't be obtained through a pure program abstraction approach. 
When mutating a seed, we modify its AST with the standard tree-edit operations (\ie node insertion, deletion and renaming), and ensure the resulted mutants also abide by the syntactic grammar of the language. 
Figure~\ref{fig:mutants} shows two mutants among many we have discovered which also preserve the predictions code2seq made for the original programs. 

\begin{figure*}[thb!]
	\centering
	\begin{subfigure}{0.46\textwidth}
		\lstset{style=mystyle}
\begin{lstlisting}[linewidth=6.7cm,morekeywords={var, public, String}]
public void saveBitmapToFile(Bitmap bmp, File file) {
    _Unknown_();
    _Type_ _var_ = 
        _Unknown_("failed to save frame"); 
}
\end{lstlisting}
	\caption{}
	\label{fig:mutanta}		
	\end{subfigure}
\,\,
	\begin{subfigure}{0.45\textwidth}
		\lstset{style=mystyle}
\begin{lstlisting}[linewidth=
		6.3cm,morekeywords={var, public, String}]
public void saveBitmapToFile(String fileStr,  Bit-
    map bitmap, View container, boolean isShare) {
    _Unknown_._Unknown_(_var_,  
        "save to file succeeded");
}
\end{lstlisting}
	\caption{}
	\label{fig:mutantb}		
	\end{subfigure}
	\caption{Two mutants we discovered among many which also preserve the predicted label \texttt{saveBitmapToFile}.}
	\label{fig:mutants}
\end{figure*}

We convey two takeaways. First, the value of individual tokens has a heavy influence on code2seq, in both mutants, the two string constants have to be present for code2seq to predict \texttt{saveBitmapToFile}. 
Second, as explained above, code2seq cares about the syntactic structure that methods exhibit. For example, in Figure~\ref{fig:mutantb}, even with the presence of \texttt{"save to file succeeded"}, the major part of the pattern code2seq looks for, it has to be passed as a parameter into a method call. Any other operation on the string (\eg \texttt{String _var_ = ``save to file succeeded''}) 
will lead code2seq to predict a different label. 



\subsubsection{Concretize}
Directly using seeds and mutants as the definition of the pattern that code2seq learned for the label \texttt{saveBitmapToFile} is a faulty approach despite their dominance in the predictions that code2seq makes.
Because no formal guarantees can be given that warrant the predicted label \texttt{saveBitmapToFile} when the seed or mutant statements are surrounded by any context of code. Therefore, we first explore the space of all full-fledged programs that are valid concretizations based on seeds and mutants, then use valid concretizations to define the pattern that code2seq learned for the label \texttt{saveBitmapToFile}. We call a concretized program valid when it preserves the predictions models made for the original program. 

To synthesize a full-fledged program, one can simply enumerate all possible statements to insert into a seed or a mutant. However, such an approach is likely to be infeasible. Because the space of programs can be enumerated by a grammar of any programming language is infinite. We show how to overcome this feasibility challenge while looking for valid concretizations in Section~\ref{subsubsec:concre}. Figure~\ref{fig:full-fledged} depicts two valid concretizations based on the mutants (Figure~\ref{fig:mutanta} and~\ref{fig:mutantb}) whose statements are highlighted in the shadow box. Apparently, both programs are written in drastically different syntax and even control flow constructs, and they do not denote the semantics of the original methods in any way, shape or form. code2seq keeps its original predictions purely because of the presence of mutant statements. The examples are convincing evidence that code2seq does not learn to represent the semantics of a method, instead, it attends to local, small syntactic features to memorize a method name.

\begin{figure*}[thb!]
	\centering
	\begin{subfigure}{0.39\textwidth}
		\lstset{style=mystyle}
\begin{lstlisting}[linewidth=5.4cm,morekeywords={var, public, String}]
public void saveBitmapToFile(Bitmap bmp, 
    File file) {
    @_Unknown_();@
    do {for (; _var_;) {
            final org.jbox2d.common.Vec2 
                center = new 
                org.jbox2d.common.Vec2(
                upperBound);
        }
        out.y = (lowerBound.y + 
            upperBound.y) * 0.5F;
    } while (t1 > t2);
    @_Type_ _var_ = @
        @_Unknown_("failed to save frame");@
}
\end{lstlisting}
		\caption{}
	\end{subfigure}
\,\,
	\begin{subfigure}{0.46\textwidth}
		\lstset{style=mystyle}
\begin{lstlisting}[linewidth=6.5cm,morekeywords={var, public, String}]
public void saveBitmapToFile(String fileStr, Bitmap
bitmap, View container, boolean isShare) {
    final Vec2 p = argPool.popVec2();
    final Vec2 d = argPool.popVec2();
    p.set(input.p1);
    d.set(input.p2).subLocal(input.p1);
    if (b.lowerBound.x - a.upperBound.x > 0.0f || 
            b.lowerBound.y - a.upperBound.y > 0.0f) 
        return;
    else if (_var_)
        avian.Utf8.cram(buf, j++, 
            ((x & 0x1f) << 6) | (y & 0x3f));
    @_Unknown_._Unknown_(_var_, @
        @"save to file succeeded");@
}
\end{lstlisting}
		\caption{}
	\end{subfigure}
	\caption{Two concretized programs for which code2seq predicts \texttt{saveBitmapToFile}.}
	\label{fig:full-fledged}
\end{figure*}


\begin{figure*}[!thbp]
	\setlength{\grammarindent}{12em} 

	\begin{grammar}
		<seed declaration> ::= <method header>  \{ <seed block statements>? \}
		
		<method header> ::= <method modifiers>? <result type> <method declarator> <throws>?
		
		<result type> ::=  boolean | byte | short | int | long | char | float | double | void
		
		<method modifiers> ::= <method modifier> | <method modifiers> <method modifier>
		
		<method modifier> ::= public | protected | private | static  
		
		<method declarator> ::= <identifier> ( <seed parameter> ) 
		
		<seed parameter>  ::= Bitmap bmp, 	File file
		\alt String file,  Bitmap bitmap, View container, boolean isShare
		
		<seed block statements> ::=	<block statements> <seed core>  <block statements>  
		
		<seed core> ::= <seed1 core>  | <seed2 core> 
		
		<seed1 core>  ::= <identifier>(); <block statements> <seed1 statement>  
		
		<seed1 statement> ::=  <ambiguous name>(‘‘failed to save frame’’);
		\alt <type> <identifier> = <ambiguous name>(‘‘failed to save frame’’);
		
		<seed2 core>  ::= <ambiguous name>( <identifier>, ‘‘save to file succeeded’’);
		\alt <ambiguous name>( <identifier>, ‘‘save to file failed’’); 
		
		<ambiguous name> ::= <identifier> | <identifier>.<identifier>
		
		<block statements> ::= <block statement> | <block statements> <block statement>
		
		<block statement> ::= <local variable declaration statement> | <statement>

	\end{grammar}
	
	%
	\caption{The context-free grammar inferred that defines the pattern code2seq learned for \texttt{saveBitmapToFile}. }
	\label{fig:grammar}
\end{figure*}

\subsubsection{Summarize}
Finally, to define the pattern code2seq learned for the label \texttt{saveBitmapToFile}, we infer a context-free grammar to describe all valid concretizations produced in the previous step. This in fact is a grammar inference problem~\cite{biermann1972survey,stevenson2014survey} where much of the success is still limited to inferring regular grammars~\cite{oncina1992inferring}. Fortunately, our problem setting is considerably simpler. That is the program to be dealt with already abide by the context-free grammar of the language they are written in, in other words, we don't need to generate production rules from the scratch but to recycle those a compiler would have used for parsing every concretized method.
Specifically, we take the union of the grammars that describe each concretized method as the definition of the pattern code2seq learned for the label \texttt{saveBitmapToFile}.

For the sake of clarity and simplicity, Figure~\ref{fig:grammar} depicts the context-free grammar that is inferred from only the concretizations of the two mutants in Figure~\ref{fig:mutants}. We use Backus–Naur form style of notation and add an extra quantifier, ``?'', which denotes zero or one occurrence of the quantified terminal or non-terminal. The key of the grammar is the production rule that describes how the non-terminal \texttt{<seed} \texttt{block} \texttt{statements>} can be replaced. Below, the replacement of non-terminal nodes \texttt{<seed1} \texttt{core>} and \texttt{<seed2} \texttt{core>} point to the specific concretization \wrt the two seeds. In particular, each production rule 
describes how the \texttt{block} \texttt{statements}, defined in the later rules, can be inserted into each seed without altering the code2seq's prediction. For the non-terminal nodes whose replacement is not defined in Figure~\ref{fig:grammar}, we reuse their production rules defined in the syntactic grammar of Java, the default input programming language of code2seq.
\section{Methodology}
\label{sec:met}

In this section, we give a detailed presentation of our approach to formalizing patterns learned by code summarization models. In particular, we describe our method ``Abstract, Mutate, Concretize, and Summarize''.

\subsection{Problem Definition}
Given a model $M$, a set of programs $\mathcal{P}$ (from the training set of $M$) for which $M$ predicts the label $L$. In this work, we aim to formalize the patterns $M$ learned from $P$. 
In other words, our formalization should define the common properties of $\mathcal{P}$ which $M$ regards as a ``trademark'' for any program it predicts the label $L$. We emphasize two points: (1) $\mathcal{P}$ is extracted from $M$'s training set, the only part of a dataset from which a model learns. During inference, models only attempt to match the learned patterns in test data, therefore we don't consider the non-training programs for studying what models have learned.
(2) the label $L$ that $M$ predicts can be incorrect for some (or even all) programs in $M$, 
regardless, $M$ has learned a pattern that can be formalized. In fact, we propose an example use of the pattern definitions based on the incorrect predictions $M$ made to improve its accuracy.

At a high-level, the way we define a pattern that code summarization models learned is to generalize about all programs that exhibit the pattern using context-free grammar (Definition~\ref{def:cfg}).
For the remainder of this section, we illustrate how to infer such a context-free grammar given a model and its training set.
To assist our exposition, we use the notations introduced above throughout this section.

\begin{definition}(\textit{Patterns})\label{def:cfg}
	Let $M$ be a model trained on a dataset $\mathcal{D}$. Let $\mathcal{P}$ (\st $\mathcal{P} \subsetneq \mathcal{D}$) be a set of programs for which $M$ predicts the label $L$. The patterns $M$ learned from $\mathcal{P}$ for predicting the label $L$ is a context-free grammar $G=<\!\!V,\Sigma,S,R\!\!>$ with non-terminals $V$, terminals $\Sigma$, a start symbol $S$, and production rules $R$, which specifies the common properties of the programs for which $M$ predicts the label $L$.
\end{definition}

\subsection{The Abstract, Mutate, Concretize, and Summarize Algorithm}

This section presents our pattern formalization algorithm. In particular, it describes the four key functional components: abstract, mutate, concretize and summarize.

\subsubsection{Abstract}
\label{subsubsec:abs}

The goal of this step is to identify a fragment of each method $P$ in $\mathcal{P}$ that captures the essence of the prediction a model made for $P$. We name such fragments \textit{seeds} (Definition~\ref{def:seed}) which satisfy the \textit{sufficient} and \textit{necessary} properties.

\begin{definition}(\textit{Seed})\label{def:seed}
Given a training method $P$ whose body consists of a set of statements $B$, and a model $M$ which predicts the label $L$ for $P$, another method\footnote{When processing each method, we do not change its header but the body, therefore we don't explicitly state the requirement that $\epsilon$'s header needs to be the same as $P$'s. The same thing applies to the later definitions and theorems. For simplicity, we consider the body of a method as a set of statements including the non-control flow statements and control flow statements whose bodies are replaced with empty blocks (\cf Section~\ref{subsec:setsta}).} $\!\epsilon$ with a body of statements $\hat{B}$ \st $\hat{B} \subseteq B$ is said to be a seed of $P$ \textit{iff} it is \textit{sufficient}, meaning, $M$ also predicts the label $L$ for $\epsilon$, and \textit{necessary}, meaning, there does not exist a method $\epsilon^\prime$ with a body $B^\prime$ \st $B^\prime \subsetneq \hat{B}$;
and $M$ also predicts the label $L$ for $\epsilon^\prime\!$.
\end{definition}

The intuition behind the sufficient property is to ensure that statements in a seed 
alone lead models to the same prediction they made for the original method.
As for the necessary property, our definition implies no subset of the seed statements possesses the same capability.
In other words, removing any statement in a seed will cause models to alter their predictions.

Definition~\ref{def:seed} does not guarantee the singularity of seeds within a method. When there happen to be multiple sets of statements in the body that satisfy the sufficient and necessary properties, a method will have multiple seeds. In fact, for the program in Figure~\ref{fig:exasavebit}/\ref{fig:newsavebit}, we depict another seed of his in Figure~\ref{fig:seed1bk}/\ref{fig:seed2bk}.

\begin{figure*}[thb!]
	\centering
	\begin{subfigure}{0.48\textwidth}
		\lstset{style=mystyle}
\begin{lstlisting}[linewidth=6.7cm,morekeywords={var, public, String}]
public void saveBitmapToFile(Bitmap bmp, File file) {        
    try {BufferedOutputStream bos = new BufferedOut-
            putStream(new FileOutputStream(file));}
    catch(IOException e) {
        Log.e(TAG, "failed to save frame", e);}
}
\end{lstlisting}
		\caption{}
		\label{fig:seed1bk}		
	\end{subfigure}
	\,\,
	\begin{subfigure}{0.45\textwidth}
		\lstset{style=mystyle}
\begin{lstlisting}[linewidth=6.6cm,morekeywords={var, public, String}]
public void saveBitmapToFile(String fileStr,  Bitmap 
    bitmap, View container, boolean isShare) { 

        SnackbarUtil.showShort(container, 
                "save to file failed");
}
\end{lstlisting}
		\caption{}
		\label{fig:seed2bk}		
	\end{subfigure}
	\caption{Two other seeds of the programs in Figure~\ref{fig:exasavebit} and~\ref{fig:newsavebit}.}
	\label{fig:seedbks}
\end{figure*}

Note that the multiplicity of seeds does not necessarily violate our assertion that models base their predictions on a small, local window of code. A common behavior models display is they feed off the most recognizable statement while receiving enough signals from the remaining statements to arrive at the original prediction. Seeds in Figure~\ref{fig:exasavebitseed} and~\ref{fig:seed1bk} are good examples of this behavior. code2seq treats \texttt{Log.e(TAG, ``failed to save frame'', e)} as the dominant statement for feature representation, however, it still need help from other statements (\eg \texttt{bos.close()} in Figure~\ref{fig:exasavebitseed} and the instantiation of a stream class within a \texttt{try catch} clause in Figure~\ref{fig:seed1bk}) to predict the label \texttt{saveBitmapToFile}. In other words, we consider the window models look into as centering around the dominate statement, and extending across a variety of supporting statements. Another case we have found where models look into separate places in a method is when multiple statements render similar program features such as the two seeds in Figure~\ref{fig:newsavebitseed} and~\ref{fig:seed2bk} which only differ by a word in a constant string. Therefore, from the perspective of feature representation, models can be deemed as attending to the same window of code within the program.


To identify seeds in a given method, we can adopt a brute-force approach to exhaust all subsets of the statements in the body. The algorithm runs in exponential time, and will incur $2^{n}\!-\!2$ (where $n$ is the number of statements in the method, and we don't need to consider an empty seed or itself, hence ``$-2$'') predictions. Although the approach can cope with methods of smaller size, it is hard to scale when the number of statements in the method increases. To address the potential scalability concern,
we present an optimization of the brute-force approach, which runs in quadratic time in average case. The optimization is designed based on the \textit{monotonicity} property (Definition~\ref{def:mono}) of models.
Before we give the formal definition of monotonicity, we explain our intuition at a high level. Since the weights of models are optimized to fit the training data, they behave differently on the unseen data, the degree to which depends on the distributions from which both datasets are drawn and the capacity of the models themselves. As a concrete piece of evidence, models are always shown to be less accurate on the test set than they are on the training set of any widely-acknowledged, well-established benchmark no matter how powerful the models are or how similar the training data is to the test data.
Since proving the property is out of the scope of this paper, we assume the monotonicity of models to which we have not found a violation through our large-scale experimentation.

\begin{definition}(\textit{Monotonicity})\label{def:mono}
	Given a model $M$, a method $P_0$ from the training set of $M$, a seed $\epsilon$ of $P_0$, a method $P_n$ outside of the training set of $M$, and the \textit{shortest} sequence of tree-edit operations $\delta_1,\delta_2,\dots,\delta_n$ that transforms the AST of $P_0$ into that of $P_n$, $M$ is said to be monotonic \textit{iff} it will make the same prediction for $P_{i}$ (the result of the application of $\delta_i$ on $P_{i-1}$ where $n > i \geq 1$) as it makes for $P_0$ if it makes the same prediction for $P_{j}$ (where $n \geq j > i$) as it makes for $P_0$. Similarly, $M$ will make a different prediction for $P_{i}$ (where $n \geq i > 1$) than it makes for $P_0$ if it makes a different prediction for $P_{j}$ (where $i > j \geq 1$) than it makes for $P_0$.
\end{definition}

\begin{algorithm}[!tbp]
 	\SetKwInOut{Input}{input}
	\SetKwInOut{Output}{output}
	
	\DontPrintSemicolon
	\SetKwProg{myproc}{procedure}{}{}
	\SetKwFunction{proc}{FindSeeds}
	\SetKwProg{myfunc}{function}{}{}
	\SetKwFunction{procc}{IsAbsent}
	\SetKwFunction{proco}{IsNecessaryOptimized}
	
	\SetKwData{Sz}{size}
	\SetKwData{Md}{method}
	\SetKwData{Mb}{method.body}	
	\SetKwData{Mh}{method.head}	
	\SetKwData{Mr}{method\_header}		
	\SetKwData{Sd}{new Method(selected\_statements, method.head)}	
	\SetKwData{Sr}{new Method(remaining\_statements, method.head)}		
	\SetKwData{Ml}{model}
	\SetKwData{Lb}{label}
	\SetKwData{Pd}{prediction}	
	\SetKwData{K}{k}	
	\SetKwData{St}{selected\_statements}		
	\SetKwData{Rs}{remaining\_statements}
	\SetKwData{Tr}{traversed\_statements}
	\SetKwData{Se}{seeds}	
	\SetKwData{Ic}{is\_continue}
	\SetKwData{Nm}{m}	
		
	\SetKw{Break}{break}

	\Input{\Md, \Ml, \Lb, \Tr}
	\Output{\Se}
	\myproc{\proc{\Md, \Ml, \Lb, \Tr, \Se}}{ 
		\Begin{			
			\Sz $\gets$ \Mb.\textit{Size}\! ()\;\label{line:getSize}
			\For{\K $\gets$ 1 \textbf{to} \Sz}{
				\St $\gets$ \textit{SelectStatements}\! (\Mb, \K, \Tr)\;\label{line:select}
				
				\Pd $\gets$ \Ml.\textit{Predict}\! (\St, \Mh)\;\label{line:predict}
				
				\If{\Pd == \Lb} {
					\Se $\gets$ \Se $\cap$ \{ \Sd \}		\\	\label{line:save}			
					\tcp{`$\setminus$' denotes the set minus}
					\Rs $\gets$ \Mb $\setminus$ \St				
					
					\tcc{IsAbsent (Line~\ref{line:startO}-\ref{line:endO}), a constant time algorithm, checks the absence of additional seeds in the remaining statements}
					\If{\textit{IsAbsent}\! (\Rs, \Mh, \Ml, \Lb)}{\label{line:check}						
						\Return{}
					}
				    \tcc{when execution hits this point, more seeds are to be discovered, which is handled in Line~\ref{line:recs}-\ref{line:rece}}
					\Ic $\gets$ $\mathit{True}$ \\
					\Break \label{line:break}				
				}
			    \Else {
			    	\Tr $\gets$ \Tr $\cap$ \{ \St \} 		 \label{line:quest}   	
				}			
			}
		
		    \If{\Ic} {
				\Nm $\gets$ \Sr  \\				\label{line:recs}
				\Return{\proc{\Nm, \Ml, \Lb, \Tr, \Se}}\label{line:rece}
		    }
		}
	}
	
%
%
	\myfunc{\procc{\Rs, \Mr, \Ml, \Lb}}{ \label{line:startO}
		\Begin{
			\Pd $\gets$ \Ml.\textit{Predict}\! (\Rs, \Mr)\;\label{line:predictOC}
			
			\If{\Pd $!\!=$ \Lb} {
				\Return{True}			
			}			
			\Else {
				\Return{False}\label{line:endO}
			} 
		}	
	}
	\caption{Find seed in a given method}
	\label{alg:seed}
\end{algorithm}

Algorithm~\ref{alg:seed} shows how to find seeds in a given method. The key is to avoid searching in a set of statements where there does not exist any seed. Whenever the \texttt{selected\_statements} is found to satisfy the sufficient property, we save it as a seed (Line~\ref{line:save}), and validate the absence of seeds in the remaining statements in the method (Line~\ref{line:check}). Thanks to the monotonicity property, the absence of seeds can be validated in constant time. That is, if the entire set of the remaining statements is not capable of leading models to the predictions they make for the original methods, none of its subsets will, hence need not be checked, in which case, the algorithm returns and its output can be accessed through \texttt{seeds}. However, if the remaining statements present more seeds, we break out of the current loop (Line~\ref{line:break}), and continue to search the additional seeds by recursively calling itself with a new method composed of the remaining statements (Line~\ref{line:recs}-\ref{line:rece}). To avoid re-attempting the same statements, we save the previous attempts into \texttt{traversed\_statements} (Line~\ref{line:quest}) which will be skipped in our quest for additional seeds in the future (Line~\ref{line:select}). We do not explicitly check the necessary property because all discovered seeds are by construction minimal. In terms of the running time, Algorithm~\ref{alg:seed} is guaranteed to find all seeds by traversing no greater than $\sum_{k=1}^{m} \Mycomb[n]{k}$ sets of statements where $n$ is the number of statements in a given method, $m$ is the number of statements in the largest seed of the method, and $\Mycomb[n]{k}$ denotes the number of combinations for $k$ objects selected out of $n$.  Since the vast majority of seeds consist of no greater than two statements, Algorithm~\ref{alg:seed} runs in quadratic time in average case.

\subsubsection{Mutate}

In this step, we set out to find similar programs to seeds which also exhibit the pattern that $M$ looks for when predicting the label $L$. In particular, we mutate the AST of each seed to produce new programs through the standard tree-edit operations: adding nodes, removing nodes, and renaming nodes. Algorithm~\ref{alg:mutant} gives the details in computing the set of all mutants according to Definition~\ref{def:mutant}. We defer the discussion on the motivation of Definition~\ref{def:mutant} to Section~\ref{subsubsec:concre}.

\begin{definition}(\textit{Mutant})\label{def:mutant}
	Given a training method $P$ (with a body $B$), its seed $\epsilon$ (with a body $\hat{B}$), and a model $M$, a method $\omega$ whose body $\bar{B}$ is a variant of $\hat{B}$ (obtained by modifying the AST of one or more statements in $\hat{B}$ with the tree-edit operations), is said to be a mutant of $\epsilon$ \textit{iff} $M$ makes the same prediction for $\epsilon$ and $\omega^\prime$ such that $\omega^\prime$'s body is $B \setminus \hat{B} \cup \bar
	{B}$. $\omega$ is said to be the weakest mutant \textit{iff} any change on $\omega$ that causes itself to be even further away from $\epsilon$ will lead $M$ to predict a different label than $M$ predicts for $\epsilon$. In contrast, $\omega$ is said to be the strongest mutant \textit{iff} it alone keeps the predictions that models made for $\epsilon$ (exemplified by the programs in Figure~\ref{fig:mutants}).
\end{definition}

Algorithm~\ref{alg:mutant} adopts an iterative approach to find mutants through a set of seeds. In each iteration, we enumerate the minimal edits to transform a mutant discovered in the previous iteration (Line~\ref{line:enu}). By minimality, we mean applying smaller edits to the corresponding mutant results in a syntactically invalid program (\eg the minimal tree-edit operations needed to transform the AST of \texttt{a+b} into that of \texttt{a[b]} is two). For the very first iteration, we transform the seeds given as input. Note that we only consider modifications when enumerating the edits for each mutant because deletions are guaranteed to yield invalid mutants according to the definition of seed; insertions are by nature not related to seed modification, and will be handled in the next step. We then check the validity of each transformed program according to the definition of mutant (Line~\ref{line:checks}-\ref{line:checke}), if the transformed program passes the validity check, we save it as a mutant, and prepare it for further modification in the later iterations (Line~\ref{line:valids}-\ref{line:valide}). On the other hand, if the transformed program failed the check, it will be permanently discarded. Because any further modification of the program will not lead to valid mutants according to the monotonicity property. The iteration stops when none of the transformed programs passes the validity check, in which case we have found all mutants.

\begin{algorithm}[!tbp]
	\SetKwInOut{Input}{input}
	\SetKwInOut{Output}{output}
	
	\DontPrintSemicolon
	\SetKwProg{myproc}{procedure}{}{}
	\SetKwFunction{proc}{FindMutants}
	\SetKwProg{myfunc}{function}{}{}
	
	\SetKwData{Md}{method}
	\SetKwData{Mb}{method.body}	
	\SetKwData{Mh}{method.head}	
	\SetKwData{Mr}{method\_header}		
	\SetKwData{Sr}{new Method(remaining\_statements, method.head)}		
	\SetKwData{Ml}{model}
	\SetKwData{Lb}{label}
	\SetKwData{Pd}{prediction}	
	\SetKwData{K}{k}	
	\SetKwData{St}{selected\_statements}		
	\SetKwData{Rs}{remaining\_statements}
	\SetKwData{Tr}{traversed\_statements}
	\SetKwData{Se}{seeds}	
	\SetKwData{Sd}{seed}		
	\SetKwData{Ic}{is\_continue}
	\SetKwData{Nm}{m}	
	\SetKwData{Mu}{mutants}	
	\SetKwData{Fu}{latest\_mutants}		
	\SetKwData{Tm}{transformed\_seeds}		
	\SetKwData{Td}{transformed\_seed}		
	\SetKwData{Tb}{transformed\_seed.body}
	
	\SetKw{Break}{break}
	
	\Input{\Md, \Ml, \Lb, \Se}
	\Output{\Mu}
	\myproc{\proc{\Md, \Ml, \Lb, \Se, \Mu}}{ 
		\Begin{			
			\Fu $\gets$ $\varnothing$ \; 
			
			\ForEach{\Sd  $\in$ \Se}{
				\Tm $\gets$ \textit{EnumerateMinimalValidEdits}\! (\Sd)\; \label{line:enu}
				\ForEach{\Td $\in$ \Tm}{
					\Rs $\gets$ \Md $\setminus$ \Sd	\;  \label{line:checks}
					\Tb $\gets$ \Tb $\cap$ \Rs \;
					
					\Pd $\gets$ \Ml.\textit{Predict}\! (\Td)\;
					\If{\Pd == \Lb} {				    \label{line:checke}	  
						\Fu $\gets$ \Fu $\cap$ \{ \Td \} \;		 \label{line:valids}	
						\Mu $\gets$ \Mu $\cap$ \{ \Td \}	\label{line:valide}	 		    		     
					}
				}
			}
			\If{\Fu $!\!=$ $\varnothing$} {
				\proc{\Md, \Ml, \Lb, \Fu, \Mu}	 		 
			}
		}
	}
	\caption{Find mutants through a set of seeds}
	\label{alg:mutant}
\end{algorithm}


\vspace*{4pt}
\noindent
\textbf{\textit{Changing Identifiers and Constants.}}\, 
Regarding the renaming operations for terminal nodes of identifiers
%
(\eg variables, types, or methods names),
they can be changed to any value that respects the lexical grammar (\eg keywords that are reserved for the language can not be used to name variables).
To fully explore the search space, we consider both swapping identifiers within a given method and changing them to 
words that do not appear in the given method. As for constants, we apply the same trick without incurring type errors (\eg \texttt{int a = `3'}). 
 
\subsubsection{Concretize}
\label{subsubsec:concre}

Even though the seeds and mutants are mostly responsible for the predictions models make, it is impromptu to take their properties as the definition of a pattern that models have learned. Because statements in the seeds (or mutants) do not guarantee to preserve the predictions models made for the original programs, considering that their surrounding context can be composed of arbitrary code.
Therefore, the goal of this step is to find the space of programs anchored by the seed (or mutant) statements that keeps the predictions that models made for the original program.

\vspace*{4pt}
\noindent
\textbf{\textit{Concretize with Seeds.}}\, As discussed in Section~\ref{subsubsec:abs}, due to the nature of the mechanism in which models get trained, the discrepancy in their performance on the training and test data is inevitable. The closer an unseen sample is to a training example, the higher probabilities it will get the same predicted label as the training example.
Based on this property, we introduce our technique to concretize a seed below. 

We simply restore the removed statements in the original method, and incrementally edit the restored statements (while leaving the seed statements intact) until the resultant method of the edits no longer keeps the prediction that models made for the seed. Our intuition is 
to quantify the space of all possible concretizations using a set of \textit{closed intervals}, each of which measures how far a concretized method is from the original method along a trajectory of edits\footnote{Like mutating seeds into mutants in Algorithm~\ref{alg:mutant}, there are multiple trajectories of edits to change non-seed statements during the concretization of a seed.}.
At the higher end of each interval is the closest concretized method (\ie the original method itself),  
which yields the widest margin for potential modifications. At the lower end is the furthest concretized method (called threshold method), albeit still underpinned by the seed statements, already exhausted the budget with its accumulated edits to the non-seed statements. In other words, any change that made itself even further from the original method will cause models to alter the prediction they made for the seed. By monotonicity, any program that lies between the two ends of each interval is a valid concretization. Recall the context-free grammar inferred for the illustrative example in Section~\ref{sec:over}, we avoid constraining the recursion depth of the production rule for the clarity of our presentation. In reality, depth of three gives a good approximation of the lower end of each interval regarding the concretization of the seeds and mutants in Figure~\ref{fig:exmethodsseed} and~\ref{fig:mutants}. Finally, as a verification mechanism, we confirm the validity of each concretized method by passing them to models for their predicted labels. A concretized method will only be kept when it lead models to the prediction they made for the original program.

Technically, when varying the non-seed statements, we first keep the control flow structure of the original method, and only modifies the non-control statements or expressions in the control construct.
Later, when switching to a different control flow structure, we reset the higher end of each interval to be $Q$, a program that is the closest to $P$ among all that employ the new control flow structure $\upsilon$. Formally, we define $Q$ in Equation~\ref{equ:ncl}:
\vspace{-3pt}
\begin{equation}\label{equ:ncl}
\begin{split}
Q & = \argmin_{W \in \mathcal{W}} \lvert P - W \rvert^{T} \\ 
\mathcal{W} & = \{W\, |\, \mathit{CF}(W)=\upsilon, \mathit{NCF}(W) \subseteq \mathit{NCF}(P) \cap \{\varnothing_s\} \}
\end{split}
\end{equation}
where $P$ is the original program from which the seed is derived. \textit{CF(W)} returns the control flow structure of $W$. \textit{NCF(W)} returns the set of non-control statements of $W$. The reason we consider $\{\varnothing_s\}$, a set of an empty statement, is to allow control statements with empty bodies when existing non-control statements are running out (\eg number of control statements is greater than the that of non-control statements). Similarly, the margin $Q$ induces will be gradually consumed by the changes we apply to $Q$. Like before, the enumeration terminates when $Q$ no longer preserves the prediction models made for the seed, and we deem programs between the two ends of each interval as the concretization of the seeds with different control flow structures.

\vspace*{4pt}
\noindent
\textbf{\textit{Concretize with Mutants.}}\, In the aforementioned approach, seed statements serve as the anchor in the method, and they are kept intact throughout the concretization process. Now we discuss how to handle the situation in which seeds statements are modified. Once a seed is modified, its statements may lose the capability of keeping the original prediction models made even with the facilitation from the non-seed statements in the original method. Figure~\ref{fig:savebitchange} gives an example, in which both programs altered code2seq's prediction due to the modification on the seed statements despite the presence of all non-seed statements. For such modified seeds which fail to preserve the original prediction models made even when the non-seed statements are present, we don't consider them for subsequent concretization. Because they do not exhibit any budget for further changes that are required to concretize a modified seed according to the monotonicity property. On the other hand, the modified seeds that remain in the mix are precisely the mutants (Definition~\ref{def:mutant}), among which the one with the least budget of change is called the weakest mutant. 

\begin{figure*}[thb!]
	\centering
	\begin{subfigure}{0.43\textwidth}
		\lstset{style=mystyle}
\begin{lstlisting}[linewidth=6.15cm,morekeywords={var, public, String}]
public void saveBitmap(Bitmap bmp, File file) {
        
  try {
      BufferedOutputStream bos = new BufferedOut-
          putStream(new FileOutputStream(file));
      bmp.compress(Bitmap.CompressFormat.PNG,
          90, bos);
            
      bos.close();
      Log.e("succeed");
  }
  catch (IOException e) {
      Log.e(TAG, @"failed"@, e);
  }

}
\end{lstlisting}
		\caption{}
		\label{fig:savebitchange1}
	\end{subfigure}
\,\,
	\begin{subfigure}{0.47\textwidth}
		\lstset{style=mystyle}
\begin{lstlisting}[linewidth=7.25cm,morekeywords={var, public, String}]
public void saveBitmap(String fileStr,  Bitmap bitmap, 
  View container, boolean isShare) { 
  File file = new File(fileStr); 
  if (isShare && file.exists()) return; 
  try { FileOutputStream fos = new FileOutputStream(file);
      boolean isCompress = bitmap.compress(Bitmap.
          CompressFormat.PNG,90, fos);
      if (isCompress)
          SnackbarUtil.showShort(container, 
              @"save succeeded"@);
      else
          SnackbarUtil.showShort(container, 
              @"save failed"@);
      fos.flush(); fos.close(); }
  catch (IOException e) { e.printStackTrace(); }
}
\end{lstlisting}
		\caption{}
		\label{fig:savebitchange2}
	\end{subfigure}
	\caption{Modifying the seed statements cause code2seq to alter its predictions despite the presence of all non-seed statements. Shallow boxes highlight the changes in the seed statements.}
	\label{fig:savebitchange}
\end{figure*}


To concretize each mutant into a full-blown method, we follow the same procedure in which we concretize a seed. 
That is, we inject to the body of a mutant the non-seed statements in the original method. The resultant method, which has the largest room for potential changes, keeps being modified until models no longer preserve the prediction they made for the seeds.
Similar to the seed concretization, statements in a mutant will never be changed, because mutants themselves already cover all the valid changes. In addition, we also employ the verification mechanism to ensure every concretization based on the mutants is valid. 

\subsubsection{Summarize}
Finally, we infer a context-free grammar to describe precisely all concretized methods produced in the previous step. We declare this grammar to be the definition of the patterns $M$ learned from $\mathcal{P}$ \wrt the label $L$. To solve this grammar inference problem, we don't reinvent the wheel but to reuse the syntactic grammar that input programs already employ. That is, we extract the production rules that a compiler would have used for parsing each concretized method before combining them into a unified grammar. Formally, we take the union of the terminals, non-terminals, and production rules extracted from each grammar. Regarding the definition of the patterns that code2seq learned for the label \texttt{saveBitmapToFile}, our inferred grammar describes 13 distinct control flow structures and can be instantiated into hundreds of programs with different syntactic structures. In other words, seeds and mutants are capable of preserving the predicted label $\texttt{saveBitmapToFile}$ most of the time when the enumerated programs are similar in size to the original method. Through the grammar we inferred, we conclude that code2seq has not learned a semantic representation for methods named \texttt{saveBitmapToFile}, instead, they memorize the methods through small, syntactic features.

\subsection{Correctness}

We show our ``abstract, mutate, concretize, and summarize'' algorithm is correct \wrt Definition~\ref{def:cfg} (Patterns). In particular, Theorem~\ref{the:soun} and~\ref{the:com} gives the soundness and completeness proof for our algorithm.


\begin{theorem}[Soundness]\label{the:soun}
	Given a model $M$ trained from a dataset $\mathcal{D}$,  
	and a grammar $G$ that the algorithm inferred as the definition of the patterns $M$ learned from $\mathcal{D}$ for predicting a label $L$, the algorithm is said to be sound \textit{iff} $G$ does not describe any program for which $M$ does not predict the label $L$.
\end{theorem}

\begin{proof}
	All concretizations of seeds or mutants are verified through $M$'s predictions. The context-free grammar that the algorithm inferred describes precisely all the valid concretizations. Therefore, by construction the theorem holds.
\end{proof}

\begin{theorem}[Completeness]\label{the:com}
	Given a model $M$ trained from a dataset $\mathcal{D}$,  
	and a grammar $G$ that the algorithm inferred as the definition of the patterns $M$ learned from $\mathcal{D}$ for predicting a label $L$, the algorithm is said to be complete \textit{iff} $G$ describes any program for which $M$ predicts the label $L$.
\end{theorem}

\begin{proof}
	Assume otherwise, there exists a program $\sigma$ such that $M$ predicts the label $L$ for $\sigma$ and $G$ does not describe $\sigma$. Regarding $\sigma$'s properties, one of the following conditions has to be met:
	\begin{enumerate}[label=(\alph*)]
		\item $\sigma$ is from the training set of $M$. Recall the concretize step in our algorithm, which uses every training program with the predicted label $L$ to enumerate the valid concretizations. In other words, all training programs are already included in the concretize step, and will be described by $G$. This contradicts the assumption. Thus, the condition is not met. 
		
		\item $\sigma$ is not from the training set of $M$.
			\begin{enumerate}[label=(\arabic*)]
				\item $\sigma$ contains the statements in a seed $\epsilon$ derived from a training program $P$. Per our assumption that $G$ does not describe $\sigma$, and thus $\sigma$ is not a valid concretization of $\epsilon$, which means $\sigma$ is not on any trajectory of edits that changes $P$ to a threshold method. By monotonicity, $M$ does not predict the label $L$ for $\sigma$, hence contradicting the assumption. 
				
				\item $\sigma$ contains the statements in a mutant derived from a seed. With the same strategy as it is adopted in (1), this condition can also be refuted.
				
				\item $\sigma$ does not contain the statements in a seed, neither does it contain the statements in a mutant. As discussed in the concretize step, $\sigma$ will not lead $M$ to predict the label $L$. The condition is also false.
			\end{enumerate}		
	\end{enumerate}	
	Since none of the conditions above can be satisfied, it can be inferred that $\sigma$ does not exist. 
\end{proof}
\section{Example Uses of the Pattern Definition}

In this section, we show the practical implications of our pattern definitions: a new method for evaluating the robustness, and a new technique to improve the accuracy of code summarization models.

\subsection{Evaluating Robustness of Code Summarization Models}
The robustness of a model refers to the reliability of the predictions it makes, especially on adversarial examples, a 
special type of data created by systematically perturbing the model inputs.~\citet{szegedy2013intriguing} is the first to discover the existence of adversarial examples in the image classification domain: visually indistinguishable perturbations cause models to alter their predictions made for the original image. Existing approaches to evaluating the robustness of a model can be classified into two categories: formally verifying a lower bound $B_l$~\cite{raghunathan2018certified,gehr2018ai2,symbolicIn} (\ie predictions are guaranteed to hold on perturbations no greater than $B_l$ \wrt some distance metric such as $L_0$, $L_1$ or $L_\infty$), or constructing attacks to demonstrate an upper bound $B_u$~\cite{goodi,CWattack} (\ie perturbations no smaller than $B_{u}$ are sufficient to make models alter their predictions). 

Wang and Christodorescu's method falls into the latter, in which they create adversarial examples by applying semantically-preserving transformations to the original programs. Interested readers are encouraged to consult the supplemental material for examples of their transformations. Despite the significant findings, their approach suffers from two issues. First, since there are many applicable transformations, and each transformation can be applied to multiple places in a given method, blindly attempting all the possibilities is quite an inefficient approach. Furthermore, their transformations often cause a fair amount of changes to the original method, thus tending to find loose bounds that do not accurately reflect the robustness of a model.

To address the weakness of their approach, we leverage the pattern definitions to pinpoint adversarial examples that demonstrate a far tighter bound than Wang and Christodorescu's method. Intuitively, since we aim to minimize the number of changes when perturbing the input methods, we only modify the seed statements, the part of the method models heavily attend to, to create adversarial examples. We introduce four semantically-preserving transformations --- variable renaming, operands swapping, API substitution, and statements reordering --- that only make minor edits to the seed of the input methods. Figure~\ref{fig:transexp} depicts an example for each transformation except variable renaming. For each column, the block at the top holds the original method and that at the bottom holds the transformed method. In comparison with Wang and Christodorescu's method, our transformations make smaller edits to the original programs.

\begin{figure*}[thb!]
	\centering
	
	\begin{subfigure}{0.29\textwidth}
		\lstset{style=mystyle}
		\begin{lstlisting}[language=C++,showstringspaces=false, morekeywords={String}]
int contains(String target,
  List<String> Items) {
  ...
  if (Items.indexOf(target)!=-1) 
  ...
}\end{lstlisting}
\end{subfigure} \vspace{-6.8pt}
\,\,
\begin{subfigure}{0.23\textwidth}
	\lstset{style=mystyle}
	\begin{lstlisting}[language=C++,showstringspaces=false]
int addChildRequest(long 
  duration) {...
  childDurationsSum =
   	childHits + duration;
  ...
}\end{lstlisting}
\end{subfigure} 
\,\,
\begin{subfigure}{0.28\textwidth}
	\lstset{style=mystyle}
	\begin{lstlisting}[language=C++,showstringspaces=false]
void addItem(List<T> mItems) {
  ...
  mIterms.add(generateItem());
  log("Item added.");
  ...
}\end{lstlisting}
\end{subfigure} 
\begin{subfigure}{0.29\textwidth}
		\lstset{style=mystyle}
		\begin{lstlisting}[language=C++,showstringspaces=false,morekeywords={String}]
int contains(String target, 
  List<String> Items) {
  ...
  if (Items.contains(target))
  ...
}\end{lstlisting}
\caption{API substitution}
\label{fig:apisub}
\end{subfigure}
\,\,
\begin{subfigure}{0.23\textwidth}
		\lstset{style=mystyle}
		\begin{lstlisting}[language=C++,showstringspaces=false]
int addChildRequest(long 
  duration) { ...
  childDurationsSum = 
		duration + childHits;
  ...
}\end{lstlisting}
\caption{Operands swapping}
\label{fig:reorder}
\end{subfigure}
	\,\,
\begin{subfigure}{0.28\textwidth}
	\lstset{style=mystyle}
	\begin{lstlisting}[language=C++,showstringspaces=false]
void addItem(List<T> mItems) {
  ...
  log("Item added.");
  mIterms.add(generateItem());
  ...
}
\end{lstlisting}
\caption{Statements reordering}
\label{fig:swap}
\end{subfigure}	
\caption{Three semantically-preserving program transformations. Top row describes the original methods, and bottom row describes the transformed methods.}
\label{fig:transexp}
\end{figure*}

\subsection{Improving Accuracy of Code Summarization Models}

We introduce a new technique based on pattern definitions that improves the prediction accuracy of code summarization models. Our technique is built on an important observation we made about the incorrect predictions that models make. When models mistakenly mis-predict a method which has a label $A$ to have a label $B$, more often than not, the seed of the mis-predicted method is similar to the seed of another method with the label $B$, in which case models are incapable of making a clear distinction, resulted in the mis-predictions. Figure~\ref{fig:mispred} gives an example in which the two methods have similar seeds but distinct names. code2seq mis-predicted the method in Figure~\ref{fig:mispred1} to be \texttt{stopAnimation}, the name of the method in Figure~\ref{fig:mispred2}.

\begin{figure*}[thb!]
	\centering
	\begin{subfigure}{0.32\textwidth}
		\lstset{style=mystyle}
\begin{lstlisting}[linewidth=4.4cm,morekeywords={var, public, String}]
public void jumpToCurrentState() {
    if (mRunningAnimator != null) {
        mRunningAnimator.end();
        @mRunningAnimator = null@;
    }
}
\end{lstlisting}
		\caption{}
		\label{fig:mispred1}
	\end{subfigure}
\,\,
	\begin{subfigure}{0.3\textwidth}
		\lstset{style=mystyle}
\begin{lstlisting}[linewidth=4.5cm,morekeywords={var, public, String}]
public void stopAnimation() {
    animating = false;
    if (animationThread != null) {
        animationThread.interrupt();
        @animationThread = null;@ }
}
\end{lstlisting}
		\caption{}
		\label{fig:mispred2}
	\end{subfigure}
	\caption{code2seq mis-predicted the name of (\subref{fig:mispred1}) to be \texttt{stopAnimation}, the name of the method in (\subref{fig:mispred2}). The reason is the seed statement in (\subref{fig:mispred1}) is similar to the seed statement in (\subref{fig:mispred2}). We highlight the seed statements of both programs within the shadow box.}
	\label{fig:mispred}
\end{figure*}

Inspired by this finding, we proposed a new approach, along the line of adversarial training~\cite{goodi}, which guides models to attend to a different part of an input method when predicting its label. The hope is the new seed will not clash with any existing seed extracted from the training programs in the entire dataset. At the technical level, we create new training programs by injecting the seed of a mis-predicted method into a variety of methods (\ie with different labels) that are correctly predicted. Each freshly-created program will be assigned with the label of the correctly predicted method which happens to host the seed. Our intention is to neutralize the previous seed, which causes models to mis-predict, to a new seed, which hopefully would lead them to the ground truth.

Figure~\ref{fig:retrain} depicts an example, in which we drop the seed statements of the mis-predicted program (Figure~\ref{fig:mispred1}) into three methods named \texttt{contains}, \texttt{count}, and \texttt{indexOf} respectively. The resulted methods will keep the name of the hosting methods as highlighted in the Figure. In practice, many such instances will be added to the training set. As a result, models are forced to shift their attention on the mis-predicted program because keep using \texttt{mRunningAnimator = null} as seed will affect their accuracy on the additional training samples despite the higher weight of the hosting methods. Because seed statements, which always emit strong signals to influence models, will not be ignored, and the only way out is to neutralize them, which is the goal of the re-training.

\begin{figure*}[htb!]
	\centering
	\begin{subfigure}{0.31\textwidth}
		\lstset{style=mystyle}
\begin{lstlisting}[linewidth=4.4cm,morekeywords={Object String}]
boolean @contains@(Object target){
   @mRunningAnimator = null;@
    
   for(Object elem: this.elements){
       if (elem.equals(target)){
           return true;
       }
   }
   return false;
}
\end{lstlisting}
		\caption{}
		\label{fig:retra1}
	\end{subfigure}
	\,
	\begin{subfigure}{0.31\textwidth}
		\lstset{style=mystyle}
\begin{lstlisting}[linewidth=4.4cm,morekeywords={String}]
int @count@(String target, ArrayList-
   <String> array){
   @mRunningAnimator = null;@

   int count = 0;
   for (String str: array)
       if (target.equals(str))
           count++;
   return count;
}
\end{lstlisting}
		\caption{}
		\label{fig:retra2}
	\end{subfigure}
	\,
	\begin{subfigure}{0.31\textwidth}
		\lstset{style=mystyle}
\begin{lstlisting}[linewidth=4.5cm,morekeywords={String}]
int @indexOf@(Object target){
   @mRunningAnimator = null;@

   int i = 0;
   for (Object elem: this.elements){
       if (elem.equals(target))
           return i;
       i++;}
   return -1;
}
\end{lstlisting}
		\caption{}
		\label{fig:retra3}
	\end{subfigure}
	\caption{Additional programs that are created to re-train code summarization models.}
	\label{fig:retrain}
\end{figure*}

To select the mis-predicted programs, we target labels where models display the highest error rates to prevent them from overfitting to a few outliers for an otherwise perfectly-predicted label. Because a poor prediction accuracy on a sizable number of programs with the same label indicates the issue of underfitting, which our new technique addresses.

\section{Evaluation}

We have realized our algorithm ``abstract, mutate, concretize and summarize'' in a tool, called \tool, 
which formalizes the patterns code summarization models learned using context-free grammar. In the first part of the evaluation, we give the details of the pattern definitions that \tool produces. For the example applications of the pattern definitions, we also evaluate the effectiveness of the method for finding adversarial examples, and the new technique for improving the prediction accuracy of code summarization models.

\subsection{Evaluation Subjects}
\vspace*{1pt}
\noindent
\textbf{\textit{Models.}}\, code2vec, code2seq, sequence GNN, LiGER are the most notable code summarization models in the literature. In the ideal case, all of them should be included in our experiments. However, we had significant difficulties in reproducing the results sequence GNN achieves~\cite{fernandes2018structured}. Our reimplementation performed significantly worse (\ie more than 10\% in F1) than theirs on the same benchmark used in their experiments. We suspect the discrepancy is caused by the inconsistent extractors which convert a method into the graph representation amendable to the model because the original extractor is the only part in their pipeline that is not open-sourced.\footnote{We have made an effort to resolve the issue with the authors of sequence GNN. Unfortunately, we could not work out a solution before the submission deadline. We have attached our email conversations in the supplemental material for reviewers' perusal.} Concerning the validity of our results on an inferior reimplementation, we regrettably exclude sequence GNN from our experiment. We also don't pick LiGER, a model that heavily depends on program executions due to its rather limited applicability and generality. Compared to code2vec and code2seq which does not even require programs to compile, LiGER requires programs to execute. For this reason, LiGER is evaluated on less than 10\% of the methods in Java-med and Java-large since the vast majority do not trigger interesting executions for LiGER to learn.

\begin{table}[!tbph]
	\captionsetup{skip=1pt}	
	\caption{Compare reimplementations to originals for code2vec and code2seq.}
	\centering
	\adjustbox{max width=\linewidth}{
		\begin{tabular}{l|ccc|ccc|ccc}
			\hline
			\multirow{2}{*}{\tabincell{c}{Models}} & \multicolumn{3}{c|}{Java-small} & \multicolumn{3}{c|}{Java-med} & \multicolumn{3}{c}{Java-large}\\
			\cline{2-10}
			& Precision & Recall & F1 & Precision & Recall & F1 & Precision & Recall & F1  \\\hline
			code2vec    & 18.51 & 18.74 & 18.62 & 38.12 & 28.31 &32.49 & 48.15 & 38.40 & 42.73    \\\hline		
			\textbf{code2vec (reimplementation)}& \textbf{19.23} & \textbf{17.72} & \textbf{18.44} & \textbf{40.32} & \textbf{28.89} &\textbf{33.66} & \textbf{48.90} & \textbf{37.26} & \textbf{42.29} \\\hline
			code2seq    & 50.64 & 37.40 & 43.02  & 61.24 & 47.07 & 53.23 & 64.03 & 55.02 & 59.19   \\\hline
			\textbf{code2seq (reimplementation)}& \textbf{48.72} & \textbf{35.46} & \textbf{41.05}  & \textbf{61.91} & \textbf{46.38} & \textbf{53.03} & \textbf{63.78} & \textbf{54.41} & \textbf{58.72}
			\\\hline
		\end{tabular}
	}
	\label{tab:com}
\end{table}

\vspace*{4pt}
\noindent
\textbf{\textit{Datasets.}}\, We use Java-small, Java-med, and Java-large, three public datasets that many code summarization models used for evaluation. They are proposed by~\citet{alon2018code2seq}, which are collections of Java methods extracted from a large number of projects on GitHub. We have re-trained code2vec and code2seq using their implementations open-sourced on GitHub. Table~\ref{tab:com} shows re-trained models are comparable to the originals~\cite{alon2018code2seq}.

\subsection{What Have code2vec and code2seq Learned}
\label{subsec:setsta}

Now, we give the details about the pattern definitions \tool produced for code2vec and code2seq.

\vspace*{4pt}
\noindent
\textbf{\textit{Finding Seeds.}}\, In general, we follow Algorithm~\ref{alg:seed} to identify the seed of a given program. Regarding the control statements, we treat their bodies to be independent of the control predicates. For example, when abstracting a \texttt{if} statement, we either delete the \texttt{if} condition and keep the statements in the body or remove a non-control statement from its body. We apply this method recursively to deal with nested control constructs. Table~\ref{tab:sizeseed} depicts the size of the seed in terms of the number of tokens it is composed of for code2vec and code2seq (\eg mean and median). The number in the parenthesis denotes the percentage a seed's tokens makes up of a whole method's. Similarly, Table~\ref{tab:stromut} gives the statistics of the strongest mutants. Apparently, both models only learned local, syntactic program features as the seeds do not capture the global, semantic properties of input methods. 

\begin{table*}[ht]
	\begin{minipage}{.49\textwidth}
		\captionsetup{skip=1pt}
		\caption{The size of the seeds.}
		\centering
		\adjustbox{max width=1\textwidth}{
			\begin{tabular}{c|cc|cc|cc}
				\hline
				\multirow{2}{*}{Models} & \multicolumn{2}{c|}{Java-small} &  \multicolumn{2}{c|}{ Java-med}  & \multicolumn{2}{c}{ Java-large}  \\
				& Mean & Median & Mean & Median & Mean & Median \\\hline
				code2vec &\tabincell{c}{14.23\\(24\%) }& \tabincell{c}{10.00\\(19\%)}  & \tabincell{c}{13.13\\(18\%)} & \tabincell{c}{9.00\\(12\%)} & \tabincell{c}{14.30\\(21\%)} & \tabincell{c}{10.00\\(21\%)}  \\\hline
				code2seq & \tabincell{c}{11.53\\(23\%)} & \tabincell{c}{9.00\\(15\%)} & \tabincell{c}{14.95\\(21\%)}  & \tabincell{c}{11.00\\(14\%)} & \tabincell{c}{18.72\\(19\%)} & \tabincell{c}{16.00\\(13\%)} \\\hline
		\end{tabular}}
		\label{tab:sizeseed}
	\end{minipage}
	\begin{minipage}{.49\textwidth}
		\captionsetup{skip=1pt}
		\caption{The size of the strongest mutants.}
		\centering
		\adjustbox{max width=1\textwidth}{
			\begin{tabular}{c|cc|cc|cc}
				\hline
				\multirow{2}{*}{Models} & \multicolumn{2}{c|}{Java-small} &  \multicolumn{2}{c|}{ Java-med}  & \multicolumn{2}{c}{ Java-large}  \\
				& Mean & Median & Mean & Median & Mean & Median \\\hline
				code2vec &\tabincell{c}{13.91\\(23\%) }& \tabincell{c}{10.00\\(19\%)}  & \tabincell{c}{12.21\\(17\%)} & \tabincell{c}{9.00\\(11\%)} & \tabincell{c}{13.70\\(20\%)} & \tabincell{c}{10.00\\(21\%)}  \\\hline
				code2seq & \tabincell{c}{12.03\\(24\%)} & \tabincell{c}{9.00\\(15\%)} & \tabincell{c}{14.32\\(21\%)}  & \tabincell{c}{11.00\\(13\%)} & \tabincell{c}{17.33\\(18\%)} & \tabincell{c}{16.00\\(13\%)} \\\hline
		\end{tabular}}
		\label{tab:stromut}
	\end{minipage}
\end{table*}

\vspace*{4pt}
\noindent
\textbf{\textit{Concretizing Seeds and Mutants.}}\, At the concretize step, we adopt the same approach to dealing with the identifiers (\eg variable, type, or method names) as we do at the mutate step. In general, concretizations of a seed or a mutant covers a wide spectrum of program structures, many of which do not resemble the semantics of the statements outside of the seed. Figure~\ref{fig:full-fledgedbks} gives two more concretizations in addition to those in Figure~\ref{fig:full-fledged}. 

\begin{figure*}[thbp!]
	\centering
	\begin{subfigure}{0.48\textwidth}
		\lstset{style=mystyle}
\begin{lstlisting}[linewidth=6.7cm,morekeywords={var, public, String}]
public void saveBitmapToFile(Bitmap bmp, File file) {
    while (A) {
        if ((x & 0xf0) == 0xe0) {
            if (!isMultiByte) {
                buf = avian.Utf8.widen(buf, j, 
                        length - 2);
                isMultiByte = true;
            }
            if (A) {
                if (!(data instanceof byte[])) 
                    return;
            }
        }
    }
    try {BufferedOutputStream bos = new BufferedOut-
            putStream(new FileOutputStream(file));}
    catch(IOException e) {
        Log.e(TAG, "failed to save frame", e);}
}
\end{lstlisting}
		\caption{}
	\end{subfigure}
	\,\,
	\begin{subfigure}{0.45\textwidth}
		\lstset{style=mystyle}
\begin{lstlisting}[linewidth=6.6cm,morekeywords={var, public, String}]
public void saveBitmapToFile(String fileStr,  Bitmap 
    bitmap, View container, boolean isShare) { 
    SnackbarUtil.showShort(_var_, 
        "save to file failed");
    if (_var_) {
        if (_var_) {
            for (; _var_;) i++;
            return avian.Utf8.trim(buf, j);
        }
    if (!(data instanceof byte[]))
        return;
    byte[] b = ((byte[]) (data));
    for (int i = 0; i < b.length; ++i) {
        if ((((int) (b[i])) & 0x80) != 0) 
            return;
    }
    return;
    }
}
\end{lstlisting}
		\caption{}
	\end{subfigure}
	\caption{Two concretized programs for seed in Figure \ref{fig:seedbks}.}
	\label{fig:full-fledgedbks}
\end{figure*}

Table~\ref{tab:numDistCFS} depicts on the number of distinct control flow structures a concretization includes. Table~\ref{tab:numSynVar} gives the same statistics \wrt the actual program instances in terms of the syntactic variations. Evidently, both models pay little to none attention to the non-seed statements in input methods as they are regularly substituted with other drastically different statements without altering the predictions models made for the original programs. 

\begin{table*}[ht]
	\begin{minipage}{.49\textwidth}
		\captionsetup{skip=1pt}
		\caption{The \# of distinct control flow structures}
		\centering
		\adjustbox{max width=1\textwidth}{
			\begin{tabular}{c|cc|cc|cc}
				\hline
				\multirow{2}{*}{Models} & \multicolumn{2}{c|}{Java-small} &  \multicolumn{2}{c|}{ Java-med}  & \multicolumn{2}{c}{ Java-large}  \\
				& Mean & Median & Mean & Median & Mean & Median \\\hline
				code2vec & 12.31& 10 & 12.79 & 10& 9.00 & 8 \\\hline
				code2seq & 12.22 & 9 & 10.43 & 8 & 8.49 & 6 \\\hline
		\end{tabular}}
		\label{tab:numDistCFS}
	\end{minipage}
	\begin{minipage}{.49\textwidth}
		\captionsetup{skip=1pt}
		\caption{The \# of syntactic variations}
		\centering
		\adjustbox{max width=1.03\textwidth}{
			\begin{tabular}{c|cc|cc|cc}
				\hline
				\multirow{2}{*}{Models} & \multicolumn{2}{c|}{Java-small} &  \multicolumn{2}{c|}{ Java-med}  & \multicolumn{2}{c}{ Java-large}  \\
				& Mean & Median & Mean & Median & Mean & Median \\\hline
				code2vec & 246.52 & 273 & 248.17 & 260 & 336.00 & 336 \\\hline
				code2seq & 260.49 & 275 & 218.46& 230 & 207.92 & 225 \\\hline
		\end{tabular}}
		\label{tab:numSynVar}
	\end{minipage}
\end{table*}

\vspace{5pt}

\begin{mdframed}[style=default]
	{\hspace{4pt} 
		Both code2vec and code2seq predominately base their predictions on a small fraction of the input methods.}
\end{mdframed}

\vspace{8pt}

\subsection{Evaluating the Robustness of code2vec and code2seq}

We construct attacks to code2vec and code2seq based on the pattern definitions \tool produces. Given an input method, we identify its seed statements, on which we apply the aforementioned semantically-preserving transformations to look for the potential adversarial examples. If a transformed program leads a model to a different prediction than it made for the original method, an adversarial example is found, and the robustness of the model can be calculated by averaging the distance between the closest adversarial examples and the original methods that are correctly predicted in a test set. We adopt the tree-edit distance (with node swapping operation) as the metric to measure the distance between programs because others like $L_0$, $L_1$ or $L_\infty$, which typically used in the setting of adversarial learning, are not suitable. Table~\ref{tab:c5} and~\ref{tab:c6} depict the robustness score for code2vec and code2seq based on programs that compile, which can be deemed as the first line of defense. Table~\ref{tab:att1} and~\ref{tab:att2} show on average how many attempts two methods take to find the closest adversarial examples. Given multiple applicable transformations, we rank them in ascending order of the distance between the resultant program after the transformation is applied and the original method. For transformations that result in programs of same distance, they will be picked randomly. The baseline refers to Wang and Christodorescu's method, which makes far more attempts than our method to find adversarial examples.

\begin{table*}[!thbp]
	\begin{minipage}{.49\textwidth}
		\captionsetup{skip=1pt}
		\caption{Robustness score for code2vec.}
		\centering
		\adjustbox{max width=1\textwidth}{
			\begin{tabular}{c|c|c|c}
				\hline
				Methods & Java-small &  Java-med  & Java-large  \\\hline
				Baseline       & 3.24 & 2.95 & 3.31 \\\hline
				\textbf{\tool} & \textbf{1.78} & \textbf{1.26} & \textbf{2.43} \\\hline
		\end{tabular}}
		\label{tab:c5}
	\end{minipage}
	\begin{minipage}{.49\textwidth}
		\captionsetup{skip=1pt}
		\caption{Robustness score for code2seq.}
		\centering
		\adjustbox{max width=1\textwidth}{
			\begin{tabular}{c|c|c|c}
				\hline
				Methods & Java-small &  Java-med  & Java-large  \\\hline
				Baseline       & 2.39          & 3.50 & 3.14 \\\hline
				\textbf{\tool} & \textbf{1.12} & \textbf{1.68} & \textbf{3.02} \\\hline
		\end{tabular}}
		\label{tab:c6}
	\end{minipage}
\end{table*}
\vspace*{-12pt}
\begin{table*}[!thbp]
	\begin{minipage}{.49\textwidth}
		\captionsetup{skip=1pt}
		\caption{On average the number of attempts needed to find the closest adversarial example for code2vec.}
		\centering
		\adjustbox{max width=1\textwidth}{
			\begin{tabular}{c|c|c|c}
				\hline
				Methods & Java-small &  Java-med  & Java-large  \\\hline
				Baseline       & 12.1 & 14.3 & 15.7 \\\hline
				\textbf{\tool} & \textbf{2.1} & \textbf{3.2} & \textbf{2.4} \\\hline
		\end{tabular}}
		\label{tab:att1}
	\end{minipage}
	\begin{minipage}{.49\textwidth}
		\captionsetup{skip=1pt}
		\caption{On average the number of attempts needed to find the closest adversarial example for code2seq.}
		\centering
		\adjustbox{max width=1\textwidth}{
			\begin{tabular}{c|c|c|c}
				\hline
				Methods & Java-small &  Java-med  & Java-large  \\\hline
				Baseline       & 15.4          & 16.9 & 22.5 \\\hline
				\textbf{\tool} & \textbf{2.6} & \textbf{2.4} & \textbf{3.7} \\\hline
		\end{tabular}}
		\label{tab:att2}
	\end{minipage}
\end{table*}

We also dive deeper into the robustness scores we obtained, and find that both methods are heavily relying on variable renaming to create adversarial examples. For a fair evaluation of the other program transformations, we exclude variable renaming and repeat the same experiment. Table~\ref{tab:c7} and~\ref{tab:c8} report the percentage of programs for which adversarial examples can not be created with the other transformations for code2vec and code2seq respectively. Using the remaining programs on which adversarial examples can be created, we report the robustness score of code2vec and code2seq in Table~\ref{tab:c9} and~\ref{tab:c10}. Results presented in Table~\ref{tab:c7}-\ref{tab:c10} suggest that our approach finds adversarial examples not only for far more programs on all datasets but also with significantly smaller edits to the original methods if the variable renaming transformation is not considered.

\begin{table*}[!htbp]
	\begin{minipage}{.49\textwidth}
		\captionsetup{skip=1pt}
		\caption{The Percentage of programs for which adversarial examples can not be created for code2vec.}
		\centering
		\adjustbox{max width=1\textwidth}{
			\begin{tabular}{c|c|c|c}
				\hline
				Methods & Java-small &  Java-med  & Java-large  \\\hline
				Baseline  & 72\% & 68\% & 64\% \\\hline
				\textbf{\tool}     & \textbf{35\%} & \textbf{27\%} & \textbf{31\%} \\\hline
		\end{tabular}}
		\label{tab:c7}
	\end{minipage}
	\begin{minipage}{.49\textwidth}
		\captionsetup{skip=1pt}
		\caption{The Percentage of programs for which adversarial examples can not be created for code2seq.}
		\centering
		\adjustbox{max width=1\textwidth}{
			\begin{tabular}{c|c|c|c}
				\hline
				Methods & Java-small &  Java-med  & Java-large  \\\hline
				Baseline  & 68\% & 64\% & 67\% \\\hline
				\textbf{\tool}     & \textbf{28\%} & \textbf{29\%} & \textbf{33\%} \\\hline
		\end{tabular}}
		\label{tab:c8}
	\end{minipage}
\end{table*}
\vspace*{-15pt}
\begin{table*}[!htbp]
	\begin{minipage}{.49\textwidth}
		\captionsetup{skip=1pt}
		\caption{Re-evaluate the robustness of code2vec.}
		\centering
		\adjustbox{max width=1\textwidth}{
			\begin{tabular}{c|c|c|c}
				\hline
				Methods & Java-small &  Java-med  & Java-large  \\\hline
				Baseline       & 12.8 & 9.40 & 15.1 \\\hline
				\textbf{\tool} & \textbf{1.40} & \textbf{1.71} & \textbf{1.88} \\\hline
		\end{tabular}}
		\label{tab:c9}
	\end{minipage}
	\begin{minipage}{.49\textwidth}
		\captionsetup{skip=1pt}
		\caption{Re-evaluate the robustness of code2seq.}
		\centering
		\adjustbox{max width=1\textwidth}{
			\begin{tabular}{c|c|c|c}
				\hline
				Methods & Java-small &  Java-med  & Java-large  \\\hline
				Baseline       & 10.3          & 14.6         & 13.4 \\\hline
				\textbf{\tool} & \textbf{1.69} & \textbf{1.37} & \textbf{2.58} \\\hline
		\end{tabular}}
		\label{tab:c10}
	\end{minipage}
\end{table*}

\vspace{2pt}

\begin{mdframed}[style=default]
	{\hspace{4pt} 
 	Both code2vec and code2seq are vulnerable to adversarial examples. Perturbing the seed statements in a given method easily causes code2vec and code2seq to alter their predictions. }
\end{mdframed}

\vspace{8pt}

\subsection{Improving the Accuracy of code2vec and code2seq}

We propose a new technique to improve the prediction accuracy of code2vec and code2seq by guiding them to correct their own mis-predictions. To prepare code2vec and code2seq for the re-training, Table~\ref{tab:statsvec} and~\ref{tab:statsseq} give the number of labels under which we pick the mis-predicted programs to fix, the average error rate that the model displays on these labels, and the number of generated programs for re-training. 

\begin{table*}[h]
	\begin{minipage}{.49\textwidth}
		\captionsetup{skip=1pt}
		\caption{Preparation for the re-training of code2vec.}
		\centering
		\adjustbox{max width=1\textwidth}{
			\begin{tabular}{l|c|c|c}
				\hline
				Mis-predictions	& Java-small &  Java-med  & Java-large  \\\hline
				\# of labels  & 73 & 442 & 1186 \\\hline
				Average error rate  & 52\% & 60\% & 70\% \\\hline
				\# of generated programs   & 6,161 & 37,774 & 97,061 \\\hline
		\end{tabular}}
		\label{tab:statsvec}
	\end{minipage}
	\begin{minipage}{.49\textwidth}
		\captionsetup{skip=1pt}
		\caption{Preparation for the re-training of code2seq.}
		\centering
		\adjustbox{max width=1\textwidth}{
			\begin{tabular}{l|c|c|c}
				\hline
				Mis-predictions & Java-small &  Java-med  & Java-large  \\\hline
				\# of labels  & 118 & 552 & 1137 \\\hline
				Average error rate  & 64\% & 52\% & 72\%\\\hline
				\# of generated programs   & 6,438 & 39,592 & 99,566 \\\hline
		\end{tabular}}
		\label{tab:statsseq}
	\end{minipage}
\end{table*}

Table~\ref{tab:ret} shows the results of code2vec and code2seq on all three datasets after the re-training. Clearly, Our technique does not just overfit the models to their training set as their accuracy on the test set has also been consistently improved, especially code2seq which now achieves the state-of-the-art results on Java-med and Java-large. On the other hand, we acknowledge the improvement is not substantial, nevertheless, we believe the technique is still a significant contribution from the following aspects. 

\begin{table}[!htbp]
	\captionsetup{skip=1pt}	
	\caption{Results of code2vec and code2seq after the re-training.}
	\centering
	\adjustbox{max width=\linewidth}{
		\begin{tabular}{l|ccc|ccc|ccc}
			\hline
			\multirow{2}{*}{\tabincell{c}{Models}} & \multicolumn{3}{c|}{Java-small} & \multicolumn{3}{c|}{Java-med} & \multicolumn{3}{c}{Java-large}\\
			\cline{2-10}
			& Precision & Recall & F1 & Precision & Recall & F1 & Precision & Recall & F1  \\\hline
			Baseline    & 18.82 & 17.14 & 17.94 & 40.95 & 28.63 & 33.70 & 47.46 & 37.19 & 41.70    \\\hline		
			code2vec    & 19.23 & 17.72 & 18.44 & 40.32 & 28.89 & 33.66 & 48.90 & 37.26 & 42.29    \\\hline		
			\textbf{code2vec (re-training)}& \textbf{22.75} & \textbf{18.59} & \textbf{20.46} & 
			\textbf{43.47} & \textbf{31.48} &\textbf{36.51} & \textbf{49.35} & \textbf{38.58} & \textbf{43.31} \\\hline\hline			
			Baseline    & 49.16 & 35.25 & 41.16  & 62.24 & 46.07 & 52.95 & 63.08 & 54.12 & 58.26   \\\hline
			code2seq    & 48.72 & 35.46 & 41.05  & 61.91 & 46.38 & 53.03 & 63.78 & 54.41 & 58.72   \\\hline
			\textbf{code2seq (re-training)}& \textbf{48.05} & \textbf{40.61} & \textbf{42.08}  & \textbf{63.24} & \textbf{47.89} & \textbf{54.50} & \textbf{66.49} & \textbf{56.43} & \textbf{61.05}
			\\\hline
		\end{tabular}
	}
	\label{tab:ret}
\end{table}

\begin{itemize}
	\item \textbf{Reliability:} All results presented in Table~\ref{tab:ret} are the average over ten separate training instances except those based on code2seq with Java-large, which take more than a week to train, is the average over five. Therefore, the higher accuracy that both models show is not random noise but a reliable improvement.

	\item \textbf{Simplicity:} As a major selling point, our technique is in nature a data augmentation approach, which does not require one to change the architecture of an existing model. By systemically augmenting the training set based on the previous mis-predictions, the improvement comes with a much lower cost than designing a new model.  
	
	\item \textbf{Effectiveness:} The baseline in Table~\ref{tab:ret} augments the training set with the same amount of additional data that are randomly selected from GitHub. The added data have the same label as the generated programs used for re-training. As shown in the table, randomly augmenting the training set does not always lead models to an improved accuracy as both models display almost the same performance as before. In addition, we give concrete evidence on the effect the re-training makes on code2seq. Given the mis-predicted program (Figure~\ref{fig:mispred1}), code2seq corrects its own mistake by expanding the seed statements (highlighted in Figure~\ref{fig:cor1}), which enables him to differentiate the method in Figure~\ref{fig:cor1} from that in Figure~\ref{fig:cor2}.

\end{itemize}

%
%

\begin{figure*}[htb!]
	\setlength{\intextsep}{-10pt}
	\captionsetup[subfigure]{aboveskip=-1pt,belowskip=-1pt}	
	\setlength{\abovecaptionskip}{2pt}
	\centering
	\begin{subfigure}{0.32\textwidth}
		\lstset{style=mystyle}
\begin{lstlisting}[linewidth=4.4cm,morekeywords={var, public, String}]
public void jumpToCurrentState() {
    if (mRunningAnimator != null) {
        @mRunningAnimator.end();@
        @mRunningAnimator = null@;
    }
}
\end{lstlisting}
		\caption{}
		\label{fig:cor1}
	\end{subfigure}
	\,\,
	\begin{subfigure}{0.3\textwidth}
		\lstset{style=mystyle}
\begin{lstlisting}[linewidth=4.5cm,morekeywords={var, public, String}]
public void stopAnimation() {
    animating = false;
    if (animationThread != null) {
        animationThread.interrupt();
        @animationThread = null;@ }
}
\end{lstlisting}
		\caption{}
		\label{fig:cor2}
	\end{subfigure}
	\caption{Given the more distinct seed statements from (\subref{fig:cor2}), code2seq correctly predicts the label of (\subref{fig:cor1}).}
	\label{fig:cor}
\end{figure*}


\vspace{-6pt}

\begin{mdframed}[style=default]
	{\hspace{4pt} The augmentation approach improves the accuracy of both code2vec and code2seq by guiding them to correct their own mis-predictions, in particular, code2seq achieves the state-of-the-art results on Java-med and Java-large on top of its existing architecture.}
\end{mdframed}


\subsection{Discussion on Threats to Validity}

A major threat to the validity of our approach is the assumption we made regarding the monotonicity property of code summarization models. Intuitively, given how machine learning models are trained to fit the training data, monotonicity is a reasonable assumption. In addition, this behavior is certainly backed up by our extensive experiments as no violation has been discovered. Nevertheless, we have not given a principle proof of the property which may not hold for the evaluated models. But even the monotonicity property is proven to be false, we believe the validity of our approach still holds to a great extent. 

First of all, the soundness of the approach is not affected since the validity of seeds, mutants, and concretizations are all verified through predictions made by the subject models, which means our primary findings --- the patterns that code2vec and code2seq learned for predicting method names --- remain to be valid, as a result, our secondary contribution regarding the applications of the pattern definitions --- a new method for evaluating the robustness and a new technique for improving the accuracy of code2vec and code2seq --- are also valid. The aspect that will be affected is the completeness of our approach, that is, the inferred grammar that defines what models learned for a label could very well miss programs for which models also predict the label. However, given the maintenance of soundness and its implications, we conclude the validity of our approach mostly holds without the monotonicity property.


\section{Related Work}
\label{sec:rel}

In this section, we survey two strands of related work: studying what models have learned and predicting the names of methods given their bodies.

\subsection{Patterns Models Have Learned}

In computer vision,~\citet{han2008bottom} propose an attribute graph grammar for parsing images with man-made objects, such as buildings, hallways, and kitchens. Their algorithm focuses on detecting exclusively the shape of rectangle, and uses six production rules to specify the various spatial relationships among the detected rectangles as the basis for object detection. Later, a similar grammar is applied to the setting of cloth modeling, an important task in human recognition and tracking, to model the wide variations of cloth configurations~\cite{chen2006composite}. A particular strength of their approach is the ability to provide insights into what patterns the algorithm recognized based on the activation of the production rules during inference.

\citet{10.1007Z} propose a visualization technique to demystify the function of intermediate feature layers and the operation of convolutional neural networks. Built upon a deconvolutional network~\cite{zeiler2011adaptive}, their technique reveals the input stimuli that excite individual feature maps at any layer in the model. It also allows one to observe the evolution of features during training and to diagnose potential problems with the model. As another contribution of the technique, they can reveal which parts of the input images are important for classification.

\citet{bau2018visualizing} present an analytic framework to visualize and understand generative adversarial networks. First, they identify the units in a layer whose featuremaps correspond to the detection of a class of objects (\eg trees). Second, they intervene within the network to switch off (or back on) the detection of the class of objects, (\eg forcing the activation of the identified units to be zero), and quantify the average causal effect of the ablation. Finally, they examine the contextual relationship between these causal object units and the background. 

\subsection{Code Summarization Models}

code2vec is the first code summarization model in a cross-projecting setting. It works by (1) decomposing the Abstract Syntax Tree (AST) of an input method into a collection of AST paths, each of which is a path between nodes in the AST, starting from one terminal, ending in another terminal, and passing through the common ancestor of both terminals; (2) aggregating the embedding learned for each AST path; and (3) predicting a probability distribution over a set of given labels based on the aggregated embedding. 
code2vec also employs the attention mechanism~\cite{vaswani2017attention,bahdanau2014neural} to assign different weights for AST paths. In other words, the embedding of the method is a weighted sum of the embedding of a path in the AST.

code2seq is another notable code summarization model in the literature. Unlike code2vec, a discriminative model in nature, code2seq adopts an encoder-decoder architecture~\cite{devlin2014fast,cho-cho2014learning} to predict method names as sequences of words. For decoding, code2seq also attends to a set of representations designed to combine each path and token representation to generate the method names. code2seq achieved then the state-of-the-art results on all Java-small, Java-med, and Java-large, three public datasets of Java methods.

In a similar vein to code2seq, sequence GNN also adopts an encoder-decoder architecture to predict method names. To encode a given method, sequence GNN relies on the coordination between a sequence and graph model. That is, recurrent neural network first learns the sequence representation of each token in a program before gated graph neural network computes the state for every node in the AST. Like code2seq, they use another recurrent neural network to generate method names as sequences of words. 

LiGER is the first model that incorporates dynamic program features to the learning processing for code summarization. Their insight is that executions, which offer direct, precise, and canonicalized representations of the program behavior, help models to generalize beyond syntactic features. On the other hand, LiGER uses symbolic features learned from source code to reduce the heavy reliance a dynamic model has on program executions since high-coverage executions are not always readily available.

\section{Conclusion}

In this paper, we present the first formal definition of the patterns that code summarization models have learned. Based on this definition, we have developed a sound algorithm for producing such pattern definitions, and a working implementation for formalizing the patterns code2vec and code2seq have learned. We found that both code2vec and code2seq heavily focus on a small, local fraction of input methods to predict their names, indicating the limited generalizability of both models about global, semantic program properties. We also present two example applications of the pattern definitions. For evaluating the robustness of code2vec and code2seq, our method uses smaller perturbations and takes far fewer attempts than prior approaches to find adversarial examples. Regarding improving the accuracy of code2vec and code2seq, the new technique we designed based on adversarial training enables code2seq to achieve the state-of-the-art results on Java-med and Java-large.

\begin{itemize}
	\item We plan to further refine and deploy our implementation to inform the users of a finer-grained explanation for a predicted label such as which expression in the method is responsible for a particular word in the output so that users can better understand the reasoning models use when making predictions. 
	\item We plan to further diagnose the code summarization models to design a general, systematic framework for improving both the accuracy and robustness of code summarization models. 
	\item We plan to apply our technique to study the patterns that are learned by other models of code for solving different tasks in programming language.
\end{itemize}

\bibliography{reference}


\end{document}